\documentclass{article}

\usepackage{PRIMEarxiv}

\usepackage{pgfplots}
\usepackage[utf8]{inputenc} 
\usepackage[T1]{fontenc}    
\usepackage{hyperref}       
\usepackage{url}            
\usepackage{booktabs}       
\usepackage{amsfonts}       
\usepackage{nicefrac}       
\usepackage{microtype}      
\usepackage{lipsum}
\usepackage{xcolor}
\usepackage{tikz-cd}
\usepackage{tikz}
\usepackage{graphicx}
\graphicspath{{media/}}     
\usepackage{amsmath,amsfonts,amssymb,amssymb,amsthm,mathrsfs}
\usepackage{pgf}
\usetikzlibrary{intersections}
\usepackage{float}
\usepgfplotslibrary{units}

\usepackage{mathrsfs}
\usepackage{selectp}

\usepackage{graphicx}
\usepackage{caption}
\usepackage{subcaption}
\usepackage{scalerel}

\usetikzlibrary{shapes.misc}
\usetikzlibrary{svg.path}
\usetikzlibrary{shapes.geometric,positioning}

\tikzset{
semi/.style={
  semicircle,
  draw,
  minimum size=2em
  }
}
\tikzset{cross/.style={cross out, draw=black, minimum size=2*(#1-\pgflinewidth), inner sep=0pt, outer sep=0pt},
cross/.default={6pt}}
\definecolor{ttffqq}{rgb}{0.2,1,0}
\definecolor{rouge}{rgb}{229,10,10}
\definecolor{fuqqzz}{rgb}{0.9568627450980393,0,0.6}
\definecolor{qqwuqq}{rgb}{0,0.39215686274509803,0}
\definecolor{ccqqqq}{rgb}{0.8,0,0}
\definecolor{qqqqff}{rgb}{0,0,1}
\definecolor{ududff}{rgb}{0.30196078431372547,0.30196078431372547,1}

\usetikzlibrary{arrows}
\newcommand{\fonction}[5]{#1:\begin{array}{l|rcl}
 & #2 & \longrightarrow & #3 \\
    & #4 & \longmapsto & #5 \end{array}}

\newcommand{\p}{{\bar{p}}}
\newcommand{\m}{{\bar{m}}}
\renewcommand{\i}{{\bar{i}}}
\renewcommand{\j}{{\bar{j}}}

\newcommand{\x}{{\bar{x}}}

\renewcommand{\a}{{\bar{a}}}
\renewcommand{\b}{{\bar{b}}}

\newcommand{\n}{\bar{n}}
\renewcommand{\Box}{\square}
\newcommand{\Diam}{\Diamond}

\newcommand{\B}{{\bar{B}}}
\newcommand{\D}{{\bar{D}}}
\newcommand{\Q}{{\mc{Q}}}
\newcommand{\R}{{\mc{R}}}
\newcommand{\C}{{\overrightarrow{C}}}

\newcommand{\wh}{\widehat}
\newcommand{\mb}{\mathbb}
\newcommand{\mc}{\mathcal}

\def\restriction#1#2{\mathchoice
             {\setbox1\hbox{${\displaystyle #1}_{\scriptstyle #2}$}
              \restrictionaux{#1}{#2}}
              {\setbox1\hbox{${\textstyle #1}_{\scriptstyle #2}$}
              \restrictionaux{#1}{#2}}
              {\setbox1\hbox{${\scriptstyle #1}_{\scriptscriptstyle #2}$}
              \restrictionaux{#1}{#2}}
              {\setbox1\hbox{${\scriptscriptstyle #1}_{\scriptscriptstyle #2}$}
              \restrictionaux{#1}{#2}}}
\def\restrictionaux#1#2{{#1\,\smash{\vrule height .8\ht1 depth .85\dp1}}_{\,#2}} 

\definecolor{ttffqq}{rgb}{0.2,1,0}
\definecolor{rouge}{rgb}{229,10,10}
\definecolor{fuqqzz}{rgb}{0.9568627450980393,0,0.6}
\definecolor{qqwuqq}{rgb}{0,0.39215686274509803,0}
\definecolor{ccqqqq}{rgb}{0.8,0,0}
\definecolor{qqqqff}{rgb}{0,0,1}
\definecolor{ududff}{rgb}{0.30196078431372547,0.30196078431372547,1}

\newtheorem{theorem}{Theorem}
\newtheorem{proposition}[theorem]{Proposition}
\newtheorem{corollary}[theorem]{Corollary}
\newtheorem{lemma}[theorem]{Lemma}

\theoremstyle{remark}
\newtheorem{remark}[theorem]{Remark}

\theoremstyle{definition}
\newtheorem{definition}[theorem]{Definition}
\newtheorem{example}[theorem]{Example}

\title{Duality for the existential fragment of first-order logic on words with numerical predicates of a fixed arity}

\author{
  Mehdi Zaidi \\
  PhD Student \\
  Laboratoire Jean Alexandre Dieudonné (LJAD) \\
  Nice 
}

\begin{document}
\maketitle

\pgfplotsset{compat=1.15}

\begin{abstract}
This article fits in the area of research that investigates the application of topological duality methods to problems that appear in theoretical computer science. 
One of the eventual goals of this approach is to derive results in computational complexity theory by studying appropriate topological objects which characterise them. 
The link which relates these two seemingly separated fields is logic, more precisely a subdomain of finite model theory known as logic on words. 
It allows for a description of complexity classes as certain families of languages, possibly non-regular, on a finite alphabet. Very few is known about the duality theory relative to fragments of first-order logic on words which lie outside of the scope of regular languages. The contribution of our work is a detailed study of such a fragment. Fixing an integer $k \geq 1$, we consider the Boolean algebra $\mc B \Sigma_1[\mc N^{u}_k]$. It corresponds to the fragment of logic on words consisting in Boolean combinations of sentences defined by using a block of at most $k$ existential quantifiers, letter predicates and uniform numerical predicates of arity $l \in \{1,...,k\}$. We give a detailed study of the dual space of this Boolean algebra, for any $k \geq 1$, and provide several characterisations of its points. In the particular case where $k=1$, we are able to construct a family of ultrafilter equations which characterise the Boolean algebra $\mc B \Sigma_1[\mc N^{u}_1]$. We use topological methods in order to prove that these equations are sound and complete with respect to the Boolean algebra we mentioned.
\end{abstract}

\keywords{General topology \and Stone duality \and Logic on words \and Descriptive complexity theory \and Distributive lattices \and Boolean algebra \and Vietoris hyperspaces \and Finite colourings \and Ultrafilter equations}

This article lies at the intersection of formal language theory and duality theoretic methods.
The contribution of this article is a deepening of the knowledge currently available on existential quantification for logic on words. The consequences of applying one layer of existential quantifier to Boolean algebras of languages defined by formulas with free first-order variables, and their counterpart at the level of topological recognisers have already been well studied in \cite{GPR}. 
On the algebraic side, we apply one layer of existential quantifier, while on the topological side, we take the Vietoris hyperspace.  Yet, unlike in the case of finite recognisers where a minimisation algorithm is available, there exists, at the moment, no well-known procedure which would allow to directly derive the minimal topological recognizer out of a given topological recogniser.
Indeed, the minimal topological recognizer corresponds to the dual space of the Boolean algebra in question, and very few concrete computations of dual spaces for fragments of logic on words which lie outside of the regular case are available. This work provides a thorough study in the case where the Boolean algebra we quantify over consists of exactly every quantifier-free formula, and we add one layer of quantification. In the case $k=1$, a characterisation of the dual space of this Boolean algebra had already been discovered in an unpublished paper of Gehrke, Krebs and Pin. This allowed for a description of a similar fragment of first-order logic on words in \cite{GKP} in terms of ultrafilter equations, that is, a family of pairs of ultrafilter on words which are sufficient to characterise it.
The idea is the following. Fix a finite alphabet $A$. A Boolean algebra of languages $\mc B$ is a subalgebra of $\mc P(A^*)$, and thus the canonical embedding provides a continuous quotient map $\pi:\beta(A^*) \twoheadrightarrow \mc S(\mc B)$. By definition, $\pi$ sends an ultrafilter $\gamma \in \beta(A^{*})$ to $\{L \in \mc B \colon L \in \gamma \}$. Therefore, an equivalent way to say that $\mc B$ satisfies the ultrafilter equation $\gamma_1 \leftrightarrow \gamma_2$, for a pair $(\gamma_1,\gamma_2) \in \beta(A^{*})^2$, is to say that $\pi(\gamma_1) = \pi(\gamma_2)$. This means that a first object one needs to study in details in order to exhibit a basis of ultrafilter equations for $\mc B$ is its dual space. 

\underline{Overview of the article:}
Fix a finite alphabet $A$. Fixing an integer $k \geq 1$, we consider the Boolean algebra $\mc B \Sigma_1[\mc N^{u}_k]$. It corresponds to the fragment of first-order logic on words consisting in Boolean combinations of sentences defined by using a block of at most $k$ existential quantifiers, letter predicates and uniform numerical predicates of arity $l \in \{1,...,k\}$. In section \ref{sect:notations}, we introduce all of the material required in order to lead our study. We give a presentation of duality theory for Boolean algebras, and the relationship between modal algebras and Vietoris hyperspaces. Then, we give an introduction to logic on words. Finally, we give more details on the notion of ultrafilter equations that we mentioned previously. In section \ref{sect:dual}, we study the Boolean algebra $\mc B \Sigma_1[\mc N^{u}_k]$ and its dual space, that we denote by $X_k$. Our contribution is a complete study of this dual space, as it should be noted that a characterisation of $X_k$ was only known in the case $k=1$. We provide several characterisations of $X_k$. First, by exploiting the duality between modal algebras and Vietoris hyperspaces, we identify $X_k$ as a certain subspace of the Vietoris hyperspaces on $\beta(\mb N^k)$ to the power $A^{k}$. Second, we follow an approach that, broadly speaking, relies on us viewing the elements of $X_k$ as ``generalized words". Fixing a finite colourings of $\mb N^k$, we explain how it is sometimes possible to define an actual finite word which corresponds to an element in $\mc V(\beta(\mb N^k))^{A^{k}}$, with respect to this finite colouring. We prove that $X_k$ consists exactly in the elements of $\mc V(\beta(\mb N^k))^{A^{k}}$ that verify this property for every single finite colouring of $\mb N^k$. In the case $k=1$, we prove that it is possible to simplify this characterisation even more. In section \ref{sect:ultraequa} and \ref{section:equasig1}, we treat the question of ultrafilter equations for $\mc B \Sigma_1[\mc N^{u}_1]$. A family of ultrafilter equations for $\mc B \Sigma_1[\mc N_{0},\mc N^{u}_1]$, the fragment obtained from $\mc B \Sigma_1[\mc N^{u}_{1}]$ by adding nullary numerical predicates, has already been introduced in \cite{GKP}. We prove that it is possible to describe a general family of ultrafilter equations which can be used to describe any of these fragments.
Our approach is based on the idea that it is possible to reformulate the ultrafilter equations in terms of a certain condition over some finite colourings of $\mb N$. By doing so, we greatly reduce the amount of combinatorics required in \cite{GKP} to prove the soundness and completeness. We use this approach to find a basis of ultrafilter equations for $\mc B \Sigma_1[\mc N^{u}_1]$, and check their soundness and completeness.

\section{Notations and preliminaries}\label{sect:notations}

This section introduces the notions of duality theory and formal language theory necessary in the following sections. The reader interested in a deeper treatment can refer to \cite{Stone}, \cite{John} and \cite{Strau}.

\subsection{Duality theory for Boolean algebras}

\begin{definition}\label{def:BA}[Boolean algebras]
A \emph{Boolean algebra} is a bounded distributive lattice such that every elements admits a complement.
For any Boolean algebras $\mc B$, $\mc B'$, a Boolean algebra homomorphism is a map $f:\mc B \to \mc B'$ such that, for any $b_1,b_2 \in \mc B$, 
$f(b_1\wedge_{\mc B} b_2)=f(b_1)\wedge_{\mc B'} f(b_2)$, $f(b_1\vee_{\mc B} b_2)=f(b_1)\vee_{\mc B'} f(b_2)$
and $f(\neg b_1)=\neg f(b_1)$. \newline We denote by $\mathbf{Bool}$ the category for which the objects are Boolean algebras and morphisms are Boolean algebra homomorphisms.
\end{definition} 

The main example of a Boolean algebra is the powerset algebra of a set $S$, $\mc P(S)$, considered with set-theoretic union, intersection, and complement. An important observation is the following: in order to reconstruct the Boolean algebra $\mc P(S)$, the data of the singletons $\{s\}$, for every $s \in S$, is actually sufficient. Abstracting this situation to any Boolean algebra leads to introducing the notion of ultrafilter. 

\begin{definition} 
Fix $\mathcal{B}$ a Boolean algebra. A \emph{filter} $\mc F$ of $\mathcal{B}$ is a non-empty subset of $\mathcal{B}$ satisfying the following properties.
\begin{itemize}
\item For every $b$ in $\mc F$, and every $b'$ in $\mathcal{B}$ such that $b \leq b'$, we have that $b'$ is in $\mc F$.
\item For every $b$ and $b'$ in $\mc F$, $b \wedge b'$ is in $\mc F$.
\end{itemize}
If $\mc F$ is a filter that is not equal to $\mathcal{B}$, we say that $\mc F$ is \emph{proper}. 
A \emph{filter basis} $\mc S$ of $\mc B$ is a non-empty family of elements of $\mc B$ such that, for every $b_1, b_2 \in \mc S$, there exists a non-bottom element $c \in \mc S$ such that $c \leq b_1 \wedge b_2$.  Every filter is completely determined by any of its filter bases.
\end{definition}

We denote by $Filt(\mc B)$ the set of filters of $\mc B$. For any set $S$, we use notation $Filt(S)$ to refer to the set of filters of the powerset algebra $\mc P(S)$, and we abusively refer to those as the filters on $S$.

\begin{example}\label{example:pcpetcofini}
We provide the two following examples of filters, which will be often used throughout the article. 
\begin{itemize}
\item For any Boolean algebra $\mc B$, and for any element $b \in \mc B$, the set \[{\uparrow}b := \{b' \in \mc B \colon b \leq b'\}\] is a filter in $\mc B$ that we refer to as the \emph{principal filter} containing $b$. 
\item For any set $S$, the set of all cofinite subsets of $S$ \[Cof(S) := \{T \subseteq S  \colon  T^{c} \text{  is finite} \}\]  is a filter that we refer to as the \emph{Fréchet filter}.
\end{itemize}
\end{example}

The fundamental idea behind Stone's construction is that, by equipping the set of all ultrafilters of $\mathcal{B}$ with an appropriate topology, it is possible to recover the Boolean algebra $\mathcal{B}$.

\begin{definition}\label{def:ultrafilterBA}
An \emph{ultrafilter} $\gamma$ of $\mathcal{B}$ is a proper filter of $\mathcal{B}$ such that, for every $b$ in $\mathcal{B}$, $b$ or $\neg b$ is in $\gamma$. We denote by $\mathcal{S}(\mathcal{B})$ the space of all ultrafilters of $\mathcal{B}$, endowed with the topology generated by the sets of the form 
\[ \widehat{b} := \{\gamma \in \mathcal{S}(\mathcal{B}) \colon b \in \gamma \},\]
for every $b$ in $\mathcal{B}$. We refer to $\mc S(\mc B)$ as the \emph{dual space} of $\mc B$.
\end{definition}

\begin{definition}[Boolean spaces] \label{def:Stonespace}
A \emph{Boolean space} is a topological space that is compact, Hausdorff, and that possesses a basis of clopen subsets. We denote by $Clop(X)$ the Boolean algebra of clopen subsets of any Boolean space $X$, equipped with union, intersection and complement of subsets. We denote by \textbf{BStone} the category for which the objects are Boolean spaces and the arrows are continuous maps.
\end{definition}

\begin{theorem}[Stone duality for Boolean algebras, \cite{Stone} Theorem 67]\label{thrm:Stoneduality}
There is a contravariant equivalence of categories between the category of Boolean algebras and the category of Boolean spaces.
\end{theorem}
In particular, for any Boolean algebra $\mathcal{B}$, the map $(\widehat{\cdot}): \mathcal{B} \to Clop(\mathcal{S}(\mathcal{B}))$ which sends any element $b$ in $\mathcal{B}$ to the clopen subset $\wh{b}$
is a Boolean isomorphism.

\subsection{The dual space of the powerset algebra}\label{subsec:remainder}

\subsubsection*{General setting}
We focus on the duality theory in the particular case where the Boolean algebra we consider in Theorem \ref{thrm:Stoneduality} is the powerset algebra of a given set. Several additional observations can be made which will be useful in the next chapters, since most of the Boolean algebras we consider are of this type.

The forgetful functor from compact Hausdorff topological spaces into all topological spaces has a left adjoint, that we denote by $\beta$, which sends a general topological space to a compact Hausdorff topological space called its \emph{\v Cech-Stone compactification}. 

The \v Cech-Stone compactification of a space $S$ can be defined by the following universal property: for any compact Hausdorff space $X$ and any continuous map $f: S \to X$, there exists a unique continuous map $g: \beta(S) \to X$ such that the following diagram commutes.

\[\begin{tikzcd}
S\arrow[r,hook,"\iota_S"]\arrow[dr,"f"']& \beta(S)\arrow[d,dashrightarrow,"g"]\\& X
\end{tikzcd}\]

In particular, for any map $f: S \to T$ there exists a unique continuous map 
\[\beta f: \beta(S) \to \beta(T)\] which extends $f$, and it is defined as follows: for any $\alpha \in \beta(S)$, 
\[ \beta f(\alpha) = \{ P \subseteq T \colon f^{-1}(P) \in \alpha \}.\]

For any set $S$, the Stone dual of the powerset algebra $\mc P(S)$ corresponds to the \emph{\v Cech-Stone compactification} of the discrete space $(S,\tau_{disc})$ (see \cite{John} III, 2.1).
Endowing $S$ with the discrete topology, $S$ can be embedded as a dense subspace of $\beta(S)$ by considering the injective map $\iota_{S}: S \to \beta(S)$ which sends any $s$ in $S$ to the principal ultrafilter ${\uparrow}\{s\}$ (we use the abusive notation ${\uparrow}s$).
We denote by $\beta(S)\setminus S$ the closed subset of $\beta(S)$ of all the free ultrafilters, that is, the ultrafilters that are not principal. We refer to this as \emph{the remainder} of $\beta(S)$, and we often use the notation $^{*}S := \beta(S) \setminus S$. 

\begin{lemma}\label{lemma:cofiniteremainder}
For any set $S$, and any ultrafilter $\alpha \in \beta(S)$, we have that $\alpha \in {^{*}S}$ if, and only if, $\alpha$ contains all cofinite sets.
\end{lemma}

\begin{proof}
Fix a set $S$ and an ultrafilter $\alpha \in \beta(S)$.
We prove the negation of this equivalence, that is, that $\alpha$ is a principal ultrafilter if, and only if, there exists a cofinite set which does not belong to $\alpha$.
For the left-to-right implication, suppose that $\alpha$ is of the form ${\uparrow}s$ for some $s \in S$. Then the set $Q_s := S \setminus \{s\}$ is cofinite, and does not belong to $\alpha$.
For the right-to-left implication, assume that there is a cofinite subset $Q$ of $S$ which does not belong to $\alpha$. Then, since $\alpha$ is an ultrafilter, $Q^{c}$, which is finite, does belong to $\alpha$. Since an ultrafilter which contains a finite set is necessarily principal, we conclude that $\alpha$ is principal.
\end{proof}

Let us say a few things about clopen subsets of $\beta(S)$. First, since $(\wh{\cdot}) : \mc P(S) \to Clop(\beta(S))$ is bijective, every clopen subset $K$ of $\beta(S)$ is of the form $\wh{Q}$, for some $Q \subseteq S$. We can also prove (\cite{Walker}, Proposition 3.13) that every clopen subset of the remainder is of the form
\[^{*}Q:=\wh{Q}\setminus Q := \{\alpha \in {^{*}S} \colon Q \in \alpha \},\] for some infinite subset $Q$ of $S$. 

\begin{lemma}\label{lemma:Booleanrulesremainder}
For any set $S$, if $Q_1$ and $Q_2$ are two infinite subsets of $S$, then the following statements hold. \newline
(1): $^{*} Q_1\subseteq {^{*}} Q_2$ if, and only if, $Q_1 \setminus Q_2$ is finite. \newline
(2): $^{*} Q_1 = {^{*}} Q_2$ if, and only if, the symmetric difference $Q_1\Delta Q_2$ is finite. \newline
(3): $^{*} Q_1 \cap {^{*}} Q_2$  is non-empty if, and only if, $Q_1 \cap Q_2$ is infinite.
\end{lemma}

\begin{proof}
See \cite{Walker}, Proposition 3.14.
\end{proof}

\subsubsection*{Modal algebra and the Vietoris functor}\label{section:modalalgebra}

In \cite{Vietoris}, Vietoris introduced a generalization of the notion of
Hausdorff metrics on any compact Hausdorff space: the so-called Vietoris hyperspace of a topological space. 

\begin{definition}[Vietoris hyperspace of a Boolean space \cite{Vietoris}]\label{def:vietoris}
For any Boolean space $X$, we denote by $\mc V(X)$ the set of closed subsets of $X$. We endow it with the topology generated by the sets of the form 

\[\Box K := \{C \in \mathcal{V}(X) \colon C \subseteq K\} \text{  and  } \Diam  K := \{C \in \mathcal{V}(X) \colon C \cap K \neq \emptyset\}  ,\]
for every clopen subset $K$ of $X$, and we refer to this topological space as the \emph{Vietoris hyperspace} of $X$.  
\end{definition}

We make a few remarks on this construction. First, observe that for any clopen subset $K$ of $X$, 
\[\Diam K = (\Box K^{c})^{c},\] therefore elements of the form $\Diam K$ can be replaced by elements of the form $(\Box K')^{c}$, where $K' = K^{c}$. Therefore, we may also define $\mc V(X)$ by taking $\{\Box K, (\Box K)^{c} \colon K \in Clop(X) \}$ as a basis.
We also note that $\Box$ is meet-preserving, while $\Diam$ is join preserving.
In particular, the family of the clopen subsets of the form
\[ \langle K,K_1,...,K_n  \rangle := \Box K \cap \bigcap_{i=1}^{n} \Diam K_i,\]
where $n \geq 1$ and $K, K_1,...,K_n$ are clopen subsets of $X$ provides a basis for the Vietoris topology.

\begin{remark}\label{remark:inclusionQcaa}
Fix a finite sequence of clopen subsets $K, K_1,...,K_n \subseteq X$.
A simple, yet important, observation is that, for every clopen $K' \subseteq X$ such that $K \subseteq K'$, we have \[\langle K,K_1,...,K_n \rangle \subseteq \mc \langle K',K_1,...,K_n \rangle.\]
The same way, for every finite sequence of clopen subsets $K'_1,...,K'_n \subseteq X$ such that, for every $i \in \{1,...,n\}$, $K_i \subseteq K'_i$, we have \[\langle K,K_1,...,K_n \rangle \subseteq \mc  \langle K,K'_1,...,K'_n \rangle.\]
\end{remark}

Because $X$ is a Boolean space, the Vietoris hyperspace of $X$ is also a Boolean space, see \cite{Ernest}, Theorem 4.9. Also, note that any clopen of $\mc V(X)$ is a compact space and thus can be written as a finite union of clopens of the form $\langle K,K_1,...,K_n  \rangle$, where $n \in \mb N$ and $K, K_1,...,K_n$ are clopen subsets of $X$.

\begin{remark}\label{remark:PUfilter}
For any Boolean algebra $\mc B$, any set $S$ and any map $f: S \to \mc V(\mc S(\mc B))$ the universal property of Cech-Stone compactification states that $f$ admits a unique continuous extension $g : \beta(S) \to \mc V(\mc S(\mc B))$, defined by sending any ultrafilter $\alpha \in \beta(S)$ to
\[g(\alpha) := \bigcap_{b \in \mc B \atop \{s \in S \colon f(s) \subseteq \wh{b} \} \in \alpha }\wh{b}.\]
This remark will be useful in order to prove Proposition \ref{prop:carcImultrafilter} 
\end{remark}

Considering $\mc B$, the Boolean algebra dual to $X$, we can understand the Vietoris hyperspace on $X$ by equipping the set of filters of $\mc B$ with an appropriate topology. Since our approach relies on a more topological understanding of problems, we chose to conduct most of our reasoning in terms of closed subsets. However, it should be noted that it is only a matter of preference, and that one could formulate all of the results from section \ref{sect:dual} involving closed subsets in terms of filters instead.

\begin{proposition}\label{prop:correspfilterclosed}
Let $\mc B$ be a Boolean algebra and $X = \mc S(\mc B)$ its dual space.
The Vietoris hyperspace of $X$ is homeomorphic to the space of filters of $\mathcal{B}$, \[\mathcal{V}(X) \simeq Filt(\mathcal{B}),\] where $Filt(\mathcal{B})$ is endowed with the topology generated by the clopen sets of the form $[b]$ and $[b]^{c}$, where for every $b \in \mathcal{B}$, \[ [b]:= \{\mc F \in Filt(\mathcal{B}) \colon b \in \mc F \}.\] 
\end{proposition}

For any set $S$, by applying Proposition \ref{prop:correspfilterclosed} to the Boolean algebra $\mc P(S)$, we obtain the following result.

\begin{corollary}\label{cor:correspfilterclosed}
For any set $S$, the Vietoris hyperspace of the \v Cech-Stone compactification of $S$ is homeomorphic to the space of filters on $S$, \[\mathcal{V}(\beta(S)) \simeq Filt(S).\] 
\end{corollary}

For any Boolean algebra $\mc B$, we denote by $C_{\mc F}$ the closed subset of $\mc S(\mc B)$ corresponding to a filter $\mc F \in Filt(\mc B)$ under this correspondence, and reciprocally, we denote by $\mc F_{C}$ the filter corresponding to a closed subset $C$. 
In the particular case where $\mc B = \mc P(S)$ for some set $S$, we have that, for any $C \in \mc V(X)$,
\[\mc F_{C} = \{Q \subseteq S \colon \forall \alpha \in C, Q \in \alpha \} = \bigcap_{\alpha \in C} \alpha\]
and for any $\mc F \in Filt(S)$,
\[C_{\mc F} =  \{ \alpha \in \beta(S) \colon \mc F \subseteq \alpha  \}.\]

\begin{example}\label{example:filtersclosedsets}
Fix a set $S$. We provide different instances of the correspondence introduced in Corollary \ref{cor:correspfilterclosed}.

\begin{itemize}
\item Fix a subset $Q$ of $S$. The filter of $S$ corresponding to $\wh{Q}$, is the principal filter containing $Q$, since
\[\mc F_{\wh{Q}} = \{P \subseteq S \colon \wh{Q} \subseteq \wh{P} \} = \{P \subseteq S \colon Q \subseteq P \} = {\uparrow} Q.\]
The closed subset of $\beta(S)$ corresponding to ${\uparrow} Q$ is $\wh{Q}$, since
\[C_{{\uparrow} Q} = \bigcap_{P \in {\uparrow} Q} \wh{P} = \bigcap_{Q \subseteq P} \wh{P} = \wh{Q}.\]

\item The closed subset of $\beta(S)$ corresponding to $Cof(S)$, the filter of all of the cofinite subsets of $S$ is the remainder of $S$, and vice-versa. 
Indeed, by Proposition \ref{prop:correspfilterclosed}, we have
\[C_{Cof(S)}:= \{\alpha \in \beta(S) \colon \forall Q \in Cof(S), Q \in \alpha \}\] 
which is equal to $^{*}S$ by Lemma \ref{lemma:cofiniteremainder}, and 
\[\mc F_{^{*}S} = \{Q \subseteq S \colon ^{*}S \subseteq \wh{Q} \} = \{Q \subseteq S \colon ^{*}S \subseteq {^{*}Q} \} \]
which is equal to $Cof(S)$ by Lemma \ref{lemma:Booleanrulesremainder} (1). 
\end{itemize}
\end{example}

The Vietoris construction can be seen as a functor $\mc V:  \textbf{BStone} \to  \textbf{BStone}$ on the category of Boolean spaces and continuous functions. Indeed, if $f: X \to Y$ is continuous, then so is the continuous map $\mc V(f): \mc V(X) \to \mc V(Y)$ which sends a closet subset $C$ of $X$ to $f(C)$. 
We would like to complete the following commutative diagram.

\[\begin{tikzcd}BA\arrow[loop left,"?"] \arrow[r, bend left=50, "\mc S"]& \mathcal{S}tone\arrow[loop right,"\mathcal{V}"] \arrow[l, bend left=50, "Clop"] \end{tikzcd}\]

In order to do so, we need to define a functor $M: \textbf{Bool} \to \textbf{Bool}$ which could be seen as the dual of $\mc V:  \textbf{BStone} \to  \textbf{BStone}$. This functor sends any Boolean algebra to what is called its corresponding formal \emph{modal algebra}. We could summarize modal algebra by saying that, just as Boolean algebras are models of classical logic, modal algebras provide models of propositional modal logic. 
The reader intested in a complete introduction to the framework of modal logic and its uses can refer to \cite{BlackburnModalLogic}, and to \cite{Yde}. More specific results about the relationships between Vietoris topology, modal logic and coalgebras are also available.

\begin{definition}[Modal algebra]
For any Boolean algebra $\mc B$, we denote by $M \mc B$ the free Boolean algebra over the set of formal generators
$\{\Box b \colon b \in \mc B \}$, with the following relations: $\Box 1 = 1$ and for every $b_1, b_2$ in $\mc B$, $\Box(b_1 \wedge b_2) = \Box b_1 \wedge \Box b_2$.
\end{definition}

What this means concretely is that $M \mc B$ can be characterized as the Boolean algebra expansion of $\mc B$ with the property that, for any meet-preserving function between Boolean algebras $h: \mc B \to \mc B'$, there is a unique Boolean algebra homomorphism $\bar{h}: \mc B \to \mc B'$ which extends $h$.
We could have defined $M \mc B$ in a similar fashion by introducing the generators $\Diam b$, for any $b \in \mc B$, with the relations $\Diam 0 = 0$ and for any $b_1, b_2$ in $\mc B$, $\Diam(b_1 \vee b_2) = \Diam b_1 \vee \Diam b_2$. For any $b \in \mc B$, the relation $\Box b = \neg (\Diam \neg b)$ holds. 
It is already fairly transparent that this mirrors the topological structure provided by the Vietoris hyperspace, on the algebraic level.

We finally explain how one can use the Vietoris construction in order to understand the dual space of the modal algebra built on a Boolean algebra $\mc B$.

\begin{proposition}[\cite{John}, Proposition 4.6]\label{prop:dualofMB}
For any Boolean algebra $\mc B$, the dual space of $M\mc B$ is homeomorphic to the Vietoris hyperspace of the dual space of $\mc B$.
\end{proposition}

An important notion we will require in section \ref{sect:dual} is the notion of content of a closed subset of $\beta(S)$. Basically, it consists in only looking at the points in the closed subset which correspond to principal ultrafilters.

\begin{definition}\label{def:contentclosedset}
For any set $S$, the content of a closed subset $C$ of $\beta(S)$ is \[Cont(C):= C \cap S.\]
\end{definition}

Note that the content of a closed subset of $\beta(S)$ may very well be empty in general.

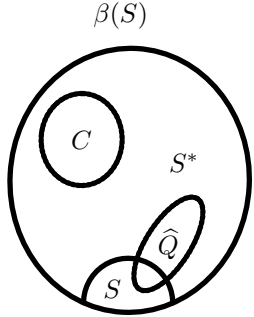
\begin{figure}[H]
\centering
\begin{tikzpicture}[line cap=round,line join=round,>=triangle 45,x=1cm,y=1cm,scale=0.56]
\clip(-9.78,-5.83) rectangle (9.78,5.83);
\draw [rotate around={86.9872124958215:(-1.29,-0.42)},line width=2pt](-1.29,-0.42) ellipse (3.1370270898428045cm and 2.8402709311626646cm);;
\draw[color=black] (80.18315820241701,80.72249301359534);
\draw [shift={(-1.3261689314691283,-3.330114168793185)},line width=2pt]  plot[domain=-0.003611629564258756:3.1379810240255344,variable=\t]({1*1.0591486560175205*cos(\t r)+0*1.0591486560175205*sin(\t r)},{0*1.0591486560175205*cos(\t r)+1*1.0591486560175205*sin(\t r)});
\node at (0,0) {$S^{*}$};
\node at (-1.7,-3) {$S$};
\node at (-1.5,3.5) {$\beta (S)$};
\draw [rotate around={88.89829388479865:(-2.43,0.55)},line width=2pt,dash pattern=on 1pt off 1pt] (-2.43,0.55) ellipse (1.0930491512347127cm and 0.9613825705799576cm)node {$C$};
\draw [rotate around={57.26477372789242:(-0.37,-1.93)},line width=2pt,dash pattern=on 1pt off 1pt] (-0.37,-1.93) ellipse (1.2775360417408788cm and 0.5242121115988744cm)node {$\wh{Q}$};
\end{tikzpicture}
\caption[The space of ultrafilters $\beta(S)$]{The space of ultrafilters $\beta(S)$. Clopen subsets must contain elements of $S$, whilst closed subsets do not need to.}
\end{figure}

For the reader who would rather prefer to reason in terms of filters, the corresponding definition is the following. The content of a filter $\mc F \in Filt(S)$, can be defined as the set of elements of $S$ appearing in every set of the filter: that is, \[Cont(\mc F) := \bigcap \mc F = \bigcap_{Q \in \mc F}Q.\]

\begin{example}
We compute the content of the closed subsets introduced in Example \ref{example:filtersclosedsets}. Fix a set $S$.
\begin{itemize}
\item For any $Q \subseteq S$, the content of the closed subset $\wh{Q}$ of $\beta(S)$ is $\wh{Q} \cap S = Q \cap S=Q$. 
\item The content of the remainder $^{*}S$ is the empty set.
This is an example of a closed subset of $\beta(S)$ with an empty content, but which is not the empty set.
More generally, for any $Q \subseteq S$, the content of the closed subset ${^{*}}Q$ of $\beta(S)$ is
$Q$.
\end{itemize}
\end{example}

\subsection{Logic on words}\label{section:logiconwords}

Throughout the rest of the article, we fix a finite alphabet $A$, and we refer to the elements of $A^{*}$, the free monoid over $A$, as the \emph{finite words} on $A$. 
We denote by $|w|$ the \emph{length} of a finite word $w = w_0...w_{|w|-1} \in A^{*}$, where for every $i \in \{0,...,|w|-1\}$, $w_i$ is in $A$. Finally, we denote by $|w|_{a}$ the number of occurrences of the letter $a$ in the word $w$. A \emph{language} $L$ is a subset of $A^{*}$, and since $\mc P(A^{*})$ is a Boolean algebra for union, intersection and complement, these operations are naturally defined on languages. 

Logic on words stems from the following idea: one way to think about a word $w$ is as a relational structure over $\{0,...,|w|-1\}$, equipped with a unary predicate $a(\cdot)$, for every $a \in A$, which allows us to tell whether the letter at a given position of $w$ is an $a$. 

\subsubsection*{Syntax}

We start by introducing the notion of numerical predicate, which will be the building blocks in order to define the formulas of logic on words. We then introduce the notion of uniformity, which, roughly speaking, will allow us to make the distinction between numerical predicates such that their interpretation takes into account the length of the words we will consider, and the ones that do not. In the next chapters, we mainly focus our study on uniform numerical predicates.

\begin{definition}[Numerical predicates]\label{numpred}
For any $k \geq 0$, a \emph{$k$-ary numerical predicate} is a map \[R^k : \mb{N}_{>0} \to \mathcal{P}(\mb{N}^k)\] such that, for all $n \geq 1$, $R^{k}(n) \subseteq \{0,...,n-1\}^k$. It is said to be \emph{uniform} if there exists a subset $Q \subseteq \mb{N}^k$ such that, for all $n \geq 1$, \[R^{k}(n) = Q \cap \{0,...,n-1\}^k.\]
\end{definition}

\begin{example}\label{example:predicates}
We now give a few examples of numerical predicates.
\begin{itemize}
\item The unary numerical predicate \[\fonction{prime}{\mb{N}_{>0}}{\mathcal{P}(\mb{N})}{n}{\{i \in \{0,...,n-1\} \colon \text{$i$ is prime} \}} \] is uniform, since for any $n \geq 1$, $prime(n)= \mathcal{P} \cap \{0,...,n-1\}$, where $\mathcal{P}$ is the set of all prime numbers. \newline
\item The binary numerical predicate \[\fonction{\leq}{\mb{N}_{>0}}{\mathcal{P}(\mb{N}^2)}{n}{\{(i,j) \in \{0,...,n-1\}^2 \colon i \leq j \}} \] is uniform, since for any $n \geq 1$, ${\leq}(n)= \mathcal{I} \cap \{0,...,n-1\}^2$, where $\mathcal{I}$ is the set of all of the couples of positive integers $(i,j)$ such that $i \leq j$. \newline
\item The unary numerical predicate \[\fonction{end}{\mb{N}_{>0}}{\mathcal{P}(\mb N)}{n}{\{n-1\} }  \] is a unary non-uniform numerical predicate. 
Indeed, there exists no subset $P \subseteq \mb N$ such that, for all $n \geq 1$, $P \cap \{0,...,n-1\}= \{n-1\}$.
\end{itemize}
\end{example}

Following the terminology introduced in \cite{Strau}, we now define the formulas of logic on words. 

\begin{definition}\label{formule}[Syntax of first-order logic on words.] We denote \emph{first-order variables} by $x,x_1,x_2$, etc.  We consider formulas that are recursively built from the following atomic blocks. \begin{itemize}
\item \emph{Letter predicates}: for every letter $a \in A$, a \emph{letter predicate} is denoted by $a(\cdot)$. For any first-order variable $x$, $a(x)$ is an atomic formula.
\item \emph{Numerical predicates}: for any $k \in \mb N$, and for any list of $k$ first-order variables $x_1,...,x_k$, if $R^{k}: \mb N_{>0} \to \mc P(\mb N^k)$ is a $k$-ary numerical predicate, then $R^{k}(x_1,...,x_k)$ is an atomic formula.
\end{itemize} 
The closure operations on formulas are the following.
\begin{itemize} 
\item If $\varphi$ and $\psi$ are formulas, then any Boolean combination of $\varphi$ and $\psi$ is a formula. 
\item If $\varphi$ is a formula, and $x$ is a variable, then $\exists x \varphi(x)$ and $\forall x \varphi(x)$ are formulas. 
\end{itemize}
\end{definition}
We say that a variable $x$ occurs \emph{freely} in a formula if it is not in the scope of a quantifier.
In particular, we call \emph{quantifier-free formulas} the Boolean combinations of atomic formulas.
A \emph{sentence} is a formula such that none of its variables are free. A \emph{fragment of first-order logic} is a subset of the set of all sentences.

\subsubsection*{Semantics}\label{sect:semantics}

As we previously announced, the particularity of logic on words is that we consider words as first-order structures. First, fix $k \in \mb N$. We use the notation $\i$ to refer to the elements $(i_1,...,i_k) \in \mb N^k$, and $\a$ to refer to the elements $(a_1,...,a_k) \in A^k$.
For any finite word $w \in A^*$, we introduce the notation 
\[|w|^{k} := \{ \i \in \mb N^k \colon \forall j \in \{1,...,k \}, i_j < |w| \} .\]
Models of formulas with free variables among $\x= \{x_1,...,x_k\}$, where all the $x_i$ are distinct, are given by elements $(w,\i)$ in \[A^{*}\otimes \mb N^k := \{(w,\i) \in A^* \times \mb N^k \colon \i \in |w|^{k} \},\] which we refer to as \emph{$\x$-structures}, one marked position in the word corresponding exactly to one free variable. Note that several variables can mark the same position. We define an equivalence relation on formulas by saying that two formulas are equivalent if they have the same models. 

We now introduce the semantic interpretation of the formulas we defined.

\begin{definition}[Semantics of logic on words]\label{def:semanticswords} We define recursively the semantics of the formulas built in Definition \ref{formule}. We start with the atomic formulas. 
\begin{itemize}
\item For any $l \in \{1,...,k\}$, and for any letter $a \in A$, the $\x$-structure $(w,\i) \in A^{*} \otimes \mb N^k$ satisfies $a(x_{i_l})$ if, and only, if $w_{i_l} = a$. 
In particular, for any $\a \in A^k$, we use the notation $w[\i] = \a$ in order to say that the $\x$-structure satisfies the formula \[ \bigwedge_{j=1}^{k}a_{j}(x_j),\] that is, for every $j \in \{1,...,k\}$, $w_{i_{j}}=a_j$. 
\item For any $l \in \{1,...,k\}$, any $j_{1},...,j_{l} \leq k$, and any $l$-ary numerical predicate $R^l$, the $\x$-structure $(w,\i) \in A^{*} \otimes \mb N^k$ satisfies $R^{l}(x_{j_{1}},...,x_{j_{l}})$ if, and only, if $(i_{j_1},...,i_{j_l})$ belongs to $R^{l}(|w|-1)$.
\end{itemize} 
Now, the closure operations on formulas are defined as follows.
\begin{itemize} 
\item The Boolean operations are interpreted in the usual way.
\item For any $l \in \{1,...,k\}$, and any $i_{\setminus l}:=(i_j)_{ 1 \leq j \leq k \atop j \neq l} \in \mb N^{k-1}$, given a formula $\varphi(\x)$, a $(\x \setminus \{x_l\})$-structure $(w,\i_{\setminus l}) \in A^{*} \otimes \mb N^{k-1}$ satisfies the formula $\exists x_l \, \varphi(\x)$ if, and only, if there exist $i_l < |w|$ such that $(w,\i)\in A^{*}\otimes \mb N^{k}$ satisfies $\varphi(\x)$. 
\end{itemize}
\end{definition}

\begin{remark}\label{remark:uniformpredicates}
In particular, if $R^{k}$ is a $k$-ary uniform numerical predicate, since there exists a subset $Q \subseteq \mb N^k$ such that, for all $n \geq 1$, $R^{k}(n) = Q \cap \{0,...,n-1\}^k$, a $\x$-structure $(w,\i) \in A^{*} \otimes \mb N^k$ satisfies $R^{k}(\x)$ if, and only if, $\i$ belongs to $Q$.
\end{remark}

\begin{example}
We introduced in Example \ref{example:predicates} the uniform binary numerical predicate $\leq$.
For any variables $x,y$, and any letters $a,b \in A$, the quantifier-free formula
\[\varphi(x,y) = a(x) \wedge b(y) \wedge {\leq}(x,y) \]
corresponds to the set $L_{\varphi(x,y)}$ of all elements $(w,i,j)$ in $A^{*} \otimes \mb N^2$ such that $w_i=a, w_j=b$ and $i \leq j$,
\[L_{\varphi(x,y)} = \{(w,i,j) \in A^* \otimes \mb N^2 \colon w[i,j]=(a,b) \text{ and } (i,j) \in \mc I  \}.\]
Applying a layer of existential quantifiers to this formula leads us to consider the sentence 
\[\exists x \exists y \; \varphi(x,y),\] which corresponds to the language consisting in all words $w \in A^{*}$ such that there exist two positions $i,j < |w|$ such that $w_i=a, w_j=b$ and $i \leq j$, in other words, the language $A^{*}aA^{*}bA^{*}$. 
\end{example}

For any formula $\varphi$ on the set of variables $\x=\{x_1,...,x_k\}$, and for any $\x$-structures $(w,\i)$, we use the notation
\[(w,\i) \models \varphi(\x) \]
to say that $(w,\i)$ satisfies the formula $\varphi$.
We denote by $L_{\varphi}$ the subset of $A^{*} \otimes \mb N^k$ of all the $\x$-structures $(w,\i)$ satisfying the formula $\varphi$,
\[ L_{\varphi} := \{(w,\i) \in A^{*} \otimes \mb N^k \colon (w,\i) \models \varphi(\x) \}.\] 
Notice that $L_{\varphi}$ is a language of finite words on the alphabet $A$ if, and only, if the formula $\varphi$ is a sentence. 
Whenever $F$ is a subset of the set of formulas with free variables among $\x$, which is closed under the Boolean connectives, the collection $\{L_{\varphi} \colon\varphi \in F\}$ is a Boolean subalgebra of $\mathcal{P}(A^{*} \otimes \mb N^{k})$. This simple observation is what allows us to apply duality for Boolean algebras to logic fragments, and motivates the approach taken in section \ref{sect:dual}.

For any $k \geq 1$, we denote by \[\mc B \Sigma_1[\mc N_k^{u}]\] the set of languages corresponding to the fragment of first order logic defined by Boolean combinations of sentences defined by only using a block of at most $k$ existential quantifiers, letter predicates and uniform numerical predicates of arity $l \in \{1,...,k\}$.

\subsection{Ultrafilter equations}\label{section:equation}

Fix a Boolean algebra of languages $\mc B \subseteq \mc P(A^{*})$. 
Its dual space $\mc S(\mc B)$ is a topological object canonically associated to $\mc B$, however, in general, it happens to be too ``big" to constitute a practical description of $\mc B$. A question that arises naturally is therefore the following: is it possible to introduce a practical topological object which holds enough information to characterize the Boolean algebra we are interested in?
The answer involves introducing the notion of \emph{ultrafilter equation}, a well-chosen family of pairs of ultrafilters in the dual space. 
A Boolean algebra of languages $\mc B$ is given by an embedding 
$\mc B \hookrightarrow \mc P(A^{*})$
and thus the dual map
$\beta(A^{*}) \twoheadrightarrow \mc S(\mc B)$
is a quotient, given by equating elements in the dual space.
The idea would be to find families of pairs of points in the dual space (ideally, much smaller than $\mc S(\mc B)^{2}$) that allow for a characterisation of $\mc B$. Since the reasoning applies for any subalgebra of a given Boolean algebra we formalize the reasoning in this setting.

\begin{definition}
For any Boolean space $X$, a \emph{Boolean equivalence relation} is an equivalence relation
$\mc E$ of $X$ such that the quotient space $X / \mc E$ is also a Boolean space.
\end{definition}

\begin{definition}\label{def:ultraequation}
For any Boolean algebra $\mc B$, any two ultrafilters $\gamma_1,\gamma_2 \in \mc S(\mc B)$, and any $b \in \mc B$, we say that $b$ satisfies the \emph{$\mc B$-equation} $\gamma_1 \leftrightarrow \gamma_2$ if, and only if, 
\[b \in \gamma_1 \Longleftrightarrow b \in \gamma_2.\]
\end{definition}

\begin{theorem}[Stone duality for Boolean subalgebras, \cite{Maitopoapproach}, Theorem 5.1]\label{prop:equaexist}
Let $\mc B$ be a Boolean algebra, and $X$ its associated dual space.
Let us consider the map from $\mc P(\mc B)$ to $\mc P(X^2)$ 
which sends any subset $S$ of $\mc B$ to
\[ \{(x,y) \in X \colon \forall b \in S, (b \in x \Longleftrightarrow b \in y) \}\]
and the map from $\mc P(X^2)$ to $\mc P(\mc B)$ 
which sends any subset $E$ of $X^2$ to
\[ \{ b \in \mc B \colon \forall (x,y) \in E, (b \in x \Longleftrightarrow b \in y) \}   .\]

These maps establish a Galois connection whose Galois closed sets are the Boolean equivalence relations on $X$ and the Boolean subalgebras of $\mc B$ respectively.
In particular, every set of equations over $X$ determines a Boolean subalgebra of 
$\mc B$, and every Boolean subalgebra of $\mc B$ is given by a set of equations over $X$.
\end{theorem}

\begin{corollary}
Any Boolean algebra of languages on a finite alphabet $A$ can be defined by a set of equations of the form $\gamma_1 \leftrightarrow \gamma_2$ where $\gamma_1$ and $\gamma_2$ are ultrafilters on the set of words.
\end{corollary}

\section{Duality for $\mc B \Sigma_1[\mc N_k^{u}]$, for any $k \geq 1$}\label{sect:dual}

\subsection{General setting} \label{section:setting}

Fix a finite alphabet $A$. We introduce for any $k \geq 1$, and any $k$-tuple of letters $\a \in A^k$ the map $c_{\a} : A^{*} \to \mc P_{fin}(\mb N^k)$, where $P_{fin}(\mb N^k)$ is the set of all finite subsets of $\mb N^k$, which sends a finite word $w$ to its \emph{{$\a$}-content},
\[c_{\a}(w):= \{ \i \in |w|^k \colon w[\i]=\a \}. \]

\begin{definition}
For any subset $Q$ of $\mb N^{k}$, we introduce the languages
\[L_{\Diam^{\a}_Q} :=  \{ w \in A^* \colon c_{\a}(w) \cap Q \neq \emptyset \}\]
and
\[L_{\Box^{\a}_Q} := (L_{\Diam^{\a}_{Q^{c}}})^{c} =   \{ w \in A^* \colon c_{\a}(w) \subseteq Q \}\]
\end{definition}

This notation allows us to keep the intuition of modal algebra, as introduced in section \ref{section:modalalgebra}.
 
\begin{example}\label{examplesuivi}
We provide a few concrete examples of these languages in the case where $k=2$, which shall be used in order to build an intuition over $\mc B_k$. Fix two letters $a$ and $b$ in $A$.
\begin{itemize} 
\item Assume that $Q$ is a subset of $\mb{N}^2$ such that there exist two subsets $P$ and $P'$ of $\mb{N}$ such that $Q = P \times P'$. In this case, we have that 
\begin{align*}
L_{\Diam^{a,b}_{P \times P'}} &= \{w \in A^{*} \mid c_{a,b}(w) \cap (P \times P') \neq \emptyset \} \\
            &= \{w \in A^{*} \mid (c_{a}(w) \times c_{b}(w)) \cap (P \times P') \neq \emptyset  \} \\
            &= \{w \in A^{*} \mid c_{a}(w) \cap P \neq \emptyset \} \cap \{w \in A^{*} \mid c_{b}(w) \cap P' \neq \emptyset \}  \\
            &= L_{\Diam^{a}_P} \cap L_{\Diam^{b}_{P'}}.
\end{align*} 
\item Assume that $Q$ is equal to $\Delta$, the diagonal of $\mb{N}^2$, in other terms every pair $(i,i)$ with $i \in \mb{N}$. Then we have
\begin{align*}
L_{\Diam^{a,a}_{\Delta}} &= \{w \in A^* \colon c_{a,a}(w) \cap \Delta \neq \emptyset \} \\
                 &=\{w \in A^* \colon \exists i \in \mb N, w_i = a \} \\
                 &= A^{*}aA^{*}.
\end{align*}
\item Assume that $Q$ is the subset $\{(i,j) \in \mb{N}^2 \mid i \leq j \}$, a similar reasoning allows us to prove that \[L_{\Diam^{a,a}_{\{(i,j) \in \mb{N}^2 \mid i \leq j \}}} = L_{\Diam^{a,a}_{\Delta}}.\]

\item Assume that $Q$ is the subset $\{(i,i+1) \mid i \in \mb{N} \}$. 
Then we have
\begin{align*}
L_{\Diam^{a,a}_{\{(i,i+1) \mid i \in \mb{N} \}}} &= \{w \in A^* \colon c_{a,a}(w) \cap \{(i,i+1) \mid i \in \mb{N} \} \neq \emptyset \} \\
                 &=\{w \in A^* \colon \exists i \in \mb N, w[(i,i+1)] = (a,a) \} \\
                 &= A^{*}aaA^{*}.
\end{align*}
\end{itemize}
\end{example} 

Let $\mc B_k$ be the Boolean subalgebra of $\mc P(A^*)$ generated by the languages $L_{\Diam^{\a}_Q}$, where $\a$ ranges over $A^k$, and $Q$ ranges over the subsets of $\mb N^k$,

\[\mc B_k := \langle \{L_{\Diam^{\a}_Q}  \colon \a \in A^k, Q \subseteq \mb N^k \} \rangle_{BA}.\]
The main purpose of this chapter is to understand, and give characterisations of the dual space of the Boolean algebra $\mc B_k$. We denote by $X_k$ the dual space of $\mc B_k$.

We denote by $V_k$ the $A^{k}$-fold power of Vietoris hyperspaces \[V_k := \mc V(\beta(\mb N^k))^{A^{k}}.\] 
We consider the function \[c^{k}:A^* \to V_k,\] which sends a finite word $w \in A^*$ to the following family of clopen subsets, \[(\wh{c_{\a}(w)})_{\a \in A^k}.\]

By the universal property of \v Cech--Stone compactification, there exists a unique continuous map $c^{k}:\beta(A^*) \to V_k$ which extends it. We give a description of the image of this map: it actually corresponds to the dual space of $\mc B_k$.

\begin{proposition}\label{prop:imck}
The image of $c^{k}$ is $X_k$.
\end{proposition}

\begin{proof}
Let us denote by $\mc M_k$ the $A^k$-fold copower of $M \mc P(\mb N^k)$, that is the Boolean algebra generated by the formal generators $\Diam^{\a} Q$, where $\a$ ranges over $A^k$, and $Q$ ranges over the subsets of $\mb N^k$
\[\mc M_k := A^k \cdot M \mc P(\mb N^k) = \langle \{ \Diam^{\a} Q \colon \a \in A^k, Q \subseteq \mb N^k \}  \rangle_{BA}. \]

As we mentioned in Proposition \ref{prop:dualofMB}, the dual of $M \mc P(\mb N^k)$ is $\mc V(\beta(\mb N^k))$. Since duality turns coproducts into products, the dual space of the Boolean algebra $\mc M_k$ is $V_k$. 
We start by defining a Boolean algebra homomorphism $h_k: \mc M_k \to \mc P(A^*)$ which we will prove to be dual to the continuous map $c^{k}: \beta(A^{*}) \to V_k$.
First, fixing $\a \in A^k$, we consider the map 
\[L_{\Diam^{\a}_{(\cdot)}}:\mc P(\mb N^k) \to \mc P(A^*)\] 
which sends any $Q \subseteq \mb N^k$ to the language $L_{\Diam^{\a}_Q}$. This map preserves finite joins: indeed, $L^{\a}_{\emptyset} = \emptyset$ and, for any $Q_1, Q_2 \subseteq \mb N^k$,
\begin{align*}
L_{\Diam^{\a}_{Q_{1}}} \cup L_{\Diam^{\a}_{Q_{2}}} 
&= \{w \in A^* \colon c_{\a}(w) \cap Q_1 \neq \emptyset \} \cup \{w \in A^* \colon c_{\a}(w) \cap Q_2 \neq \emptyset \} \\
&= \{w \in A^* \colon c_{\a}(w) \cap Q_1 \neq \emptyset \text{ or } c_{\a}(w) \cap Q_2 \neq \emptyset  \}   \\
&= \{w \in A^* \colon c_{\a}(w) \cap (Q_1 \cup Q_2) \neq \emptyset  \}   \\
&= L_{\Diam^{\a}_{(Q_{1} \cup Q_{2})}}.
\end{align*}
Therefore, this join preserving map extends uniquely to a Boolean algebra homomorphism 
\[ h^{\a}: M \mc P(\mb N^k) \to \mc P(A^*).\]
We now define $h_k: \mc M_k \to \mc P(A^*)$ by using the universal property of the $A^k$-fold copower of $M \mc P(\mb N^k)$, that is $h_k$ is the unique Boolean algebra homomorphism such that, for any $\a \in A^k$, and for any $Q \subseteq \mb N^k$,
\[h_k(\Diam^{\a}Q) = L_{\Diam^{\a}_{Q}}.\] 
In particular, this equality proves that $Im(h_{k}) = \mc B_k$, since $\mc M_k$ is the Boolean algebra generated by the elements of the form $\Diam^{\a}Q$, where $\a$ ranges over $A^k$, and $Q$ ranges over the subsets of $\mb N^k$. Therefore, we have the following commutative diagram in $\mathbf{B}$ool.
\[\begin{tikzcd}
\mc{M}_k \arrow[r,"h_k"]\arrow[dr,twoheadrightarrow] & \mc{P}(A^{*}) \\
&  \mc{B}_k \arrow[u,hookrightarrow]
\end{tikzcd}\]

Now, by duality, we have the following diagram in $\mathbf{S}$tone.

\[\begin{tikzcd}
\beta(A^*) \arrow[r,"(h_k)^{-1}"]\arrow[dr,twoheadrightarrow] & V_k \arrow[d,hookleftarrow] \\
&  X_k 
\end{tikzcd}\]

In order to conclude that $Im(c^k) = X_k$, it is enough to prove that $c^k$ is dual to $h_k$, that is that $c^k = (h_{k})^{-1}$. Since we consider continuous maps between compact Hausdorf spaces, and $A^{*}$ is a dense subspace of $\beta(A^*)$, we only need to prove that the restriction of these maps to $A^{*}$ are equal. Now, for any word $w \in A^*$, we have by definition \[c^{k}(w) = (c_{\a}(w))_{\a \in A^k},\] and since the duality turns coproducts into products, \[(h_k)^{-1}(w)= ((h^{\a})^{-1}(w))_{\a \in A^k}.\] Finally, for every $w \in A^*$, and every $\a \in A^k$,

\begin{align*}
(h^{\a})^{-1}(w)  
&=  \{\i \in \mb N^k \colon w \in h^{\a}(\Diam^{\a}({\{\i\}})) \} \\
&=  \{\i \in \mb N^k \colon w \in L_{\Diam^{\a}_{\{\i\}}} \} \\
&= \{\i \in \mb N^k \colon c_{\a}(w) \cap {\{\i\}} \neq \emptyset \}    \\
&= \{\i \in \mb N^k \colon w[\i]= \a \}  \\
&=c_{\a}(w),
\end{align*}

and we conclude that $c^k = (h_{k})^{-1}$, and therefore that $Im(c^k) = X_k$.
\end{proof} 

\subsection{Logical description of $\mc B_k$} \label{sec:chap2logicdescription}

Fix $k \geq 1$. 
In this section, we give a description of the Boolean algebra $\mc B_k$ in the context of logic on words. The Boolean algebra $\mc{B}_k$ contains the subalgebra generated by the languages $L_{\Diam^{a}_{P}}$, where $a$ ranges over $A$, and $P$ ranges over the subsets of $\mb N$. This implies that it already encodes at the very least all of the first-order sentences built by using unary uniform numerical predicates (see \cite{GKP}, Theorem 2.9 for a proof). We will prove that $\mc B_k$ can actually be identified with the Boolean algebra of languages corresponding to formulas which are Boolean combinations of sentences defined by using a block of at most $k$ existential quantifiers, letter predicates and uniform numerical predicates of arity $l \in \{1,...,k\}$.

Fixing a set of free variables $\x = \{x_1,...,x_k \}$, we first express any quantifier-free formula written by using only a subset of these variables as a normal form which involves exactly all of the free-variables in $\x$. 

\begin{lemma}\label{lemma:freeformula}
Any quantifier-free formula $\varphi$ such that the set of its variables $\{x_{j_{1}},...,x_{j_{l}}\}$ is a subset of $\x$ can be written as a formula of the form
\[\bigvee_{\a \in A^k}(\a(\x) \wedge R^{\a}(\x)),\] where, for every $\a \in A^k$, $\a(\x)$ is defined as the disjunction $\bigwedge_{j=1}^{k}a_{j}(x_j)$, and $R^{\a}$ is a $k$-ary numerical predicate.
\end{lemma}

\begin{proof}
We denote by $\iff$ the relation of logical equivalence between formulas. We prove the statement by structural induction on quantifier-free formulas. 
We start with atomic formulas. 
For letter predicates, we have that, for any $j \in \{1,...,k\}$, and any $a \in A$,
\[a(x_j)  \iff \bigvee_{\a \in A^{k}} (\a(\x)\wedge S^{\a}(\x)),\] where, for every $\a \in A^k$, $S^{\a}$ is the $k$-ary numerical predicate defined as 
\[\fonction{S^{\a}}{\mb N_{>0}}{\mc P(\mb N^k)}{n}{ \begin{cases} \{0,...,n-1\}^k \text{  if $a_j = a$}  \\ \emptyset \text{  otherwise} \end{cases}}.\]
For numerical predicates, we have that, for any $l$-ary numerical predicate $R^{l}$,
\[R^{l}(x_{j_{1}},...,x_{j_{l}}) \iff \bigvee_{\a \in A^k} (\a(\x) \wedge T_{R^{l}}^{\a}(\x)),\]
where, for every $\a \in A^k$, $T_{R^{l}}^{\a}$ is $T_{R^{l}}$, the $k$-ary numerical predicate defined as 
\[\fonction{T_{R^{l}}}{\mb N_{>0}}{\mc P(\mb N^k)}{n}{\{(n_1,...,n_k) \in \{0,...,n-1\}^k \mid (n_{j_{1}},...,n_{j_{l}}) \in R^{l}(n) \}  }.\]
To conclude, all we need to do is to prove that formulas of the form $\bigvee_{\a \in A^k} (\a(\x) \wedge R^{\a}(\x))$ are closed under Boolean operations. For any formulas $\varphi_1(\x)$ of the form $\bigvee_{\a \in A^k} (\a(\x) \wedge R_1^{\a}(\x))$ and $\varphi_2(\x)$ of the form $\bigvee_{\a \in A^k} (\a(\x) \wedge R_2^{\a}(\x))$, we have that \[\varphi_1(\x) \wedge \varphi_2(\x)  \iff \bigvee_{\a \in A^k} (\a(\x) \wedge (R_1^{\a} \cap R_2^{\a})(\x) \] 
where, for every $\a \in A^k$, $(R_1^{\a} \cap R_2^{\a})$ is the $k$-ary numerical predicate defined as
\[\fonction{R_1^{\a} \cap R_2^{\a}}{\mb N_{>0}}{\mc P(\mb N^k)}{n}{R_1^{\a}(n) \cap R_2^{\a}(n)},\] 
and \[\varphi(\x) \vee \psi(\x)  \iff \bigvee_{\a \in A^k} (\a(\x) \wedge (R_1^{\a} \cup R_2^{\a})(\x)).\] 
where, for every $\a \in A^k$, $(R_1^{\a} \cup R_2^{\a})$ is defined in an analogous way. Finally, the negation of any atomic formula is equivalent to a disjunction of atomic formulas: for any $l$-ary predicate $R^{l}$, and any free-variables $(x_{j_{1}},...,x_{j_{l}}),$ \[\neg R^{l}(x_{j_{1}},...,x_{j_{l}}) \iff ((R^{l})^{c})(x_{j_{1}},...,x_{j_{l}}),\]
where, for every $\a \in A^k$, $(R^{l})^{c}$ is the $k$-ary numerical predicate defined as
\[\fonction{(R^{l})^{c}}{\mb N_{>0}}{\mc P(\mb N^k)}{n}{ \{0,...,n-1\}^k \setminus T_{R^{l}}(n)};\] 
and for any letter $a$ and any $j \in \{1,...,k\}$,  \[\neg a(x_{j}) \iff\bigvee_{b \in A \atop b\neq a } b(x_{j}).\]
\end{proof}

Therefore, by Lemma \ref{lemma:freeformula}, if we apply one layer of existential quantifier to a quantifier-free formula $\varphi$ whose set of variables is contained in $\x$, we obtain a sentence of the form
\[\exists x_1... \exists x_k  \bigvee_{\a \in A^k}(\a(\x) \wedge R^{\a}(\x))  \ \]
where, for every $\a \in A^k$, $R^{\a}$ is a $k$-ary numerical predicate.

From now on, we will restrict our attention to uniform numerical predicates.
By Remark \ref{remark:uniformpredicates}, this means that, for every $\a \in A^k$, considering the $k$-ary numerical predicate $R^{\a}$ is equivalent to considering a subset $Q^{\a} \subseteq \mb N^k$.
Under this assumption, we denote by $\mc{B}\Sigma_1[\mc{N}^u_{k}]$, the Boolean algebra generated by the languages $L_{\psi}$, where
\[ \psi := \exists x_1... \exists x_k  \bigvee_{\a \in A^k}(\a(\x) \wedge Q^{\a}(\x)).  \ \]
Since the existential quantifier commutes with finite disjunctions, this Boolean algebra is generated by the languages corresponding to sentences of the form
\[ \bigvee_{\a \in A^k} \exists x_1... \exists x_k \; \a(\x) \wedge Q^{\a}(\x).\]

Now, note that, for any $\a \in A^k$, and any $Q^\a \subseteq \mb N^k$, the sentence
\[ \exists x_1... \exists x_k \; \a(\x) \wedge Q^{\a}(\x) \]
corresponds to the language of the form
\[ \{w \in A^* \colon \exists \i \in Q^\a \cap |w|^k, w[\i] = \a \}
= \{w \in A^* \colon c_{\a}(w) \cap Q^\a \neq \emptyset \}  =
 L_{\Diam^{\a}_{Q^\a}},\]
and thus $\mc{B}\Sigma_1[\mc{N}^u_{k}]$ is the Boolean algebra generated by the languages $\bigcup_{ \a \in A^k} L_{\Diam^{\a}_{Q^\a}}$, where $(Q_\a)$ ranges over $P(\mb N^k)^{A^{k}}$. We conclude by proving that this Boolean algebra actually coincides with $\mc B_k$.

\begin{proposition} \label{preuve1}
For any $k \geq 1$, the Boolean algebra $\mc{B}\Sigma_1[\mc{N}_{k}^{u}]$ is equal to $\mc B_k$. 
\end{proposition}

\begin{proof}
The Boolean algebra $\mc{B}\Sigma_1[\mc{N}_{k}^{u}]$ is generated by the languages $\bigcup_{ \a \in A^k} L_{\Diam^{\a}_{Q^\a}}$, where $(Q_\a)$ ranges over $\mc P(\mb N^k)^{A^{k}}$, while $\mc B_k$ is generated by the languages $L_{\Diam^{\a}_{Q}}$, where $\a$ ranges over $A^{k}$ and $Q$ ranges over subsets of $\mb N^k$.
On the one hand, it is clear that $\mc{B}\Sigma_1[\mc{N}_{k}^{u}]$ is generated by Boolean combinations of languages of the form $L_{\Diam^{\a}_{Q}}$. 
On the other hand, if we fix $Q \subseteq \mb N^k$ and a $k$-tuple of letters $\a \in A^k$, then we have that  
\[L_{\Diam^{\a}_Q} := L_{\Diam^{\a}_Q} \cup \bigcup_{\b \in A^{k} \atop \b \neq \a}L_{\Diam^{\b}_{\emptyset}},\]
and we conclude that $\mc B_k$ is isomorphic to  $\mc{B}\Sigma_1[\mc{N}_{k}^{u}]$.
\end{proof}

\subsection{The dual space via finite colourings}\label{section:colouring}

In this section, we provide the first elements of study of the dual space $X_k$, for any $k \geq 1$.
We start by explaining how it is possible to make an analogy between elements of the dual space and finite words. Formalizing this link, and considering a different basis of $V_k$ constructed out of the family of all finite colourings of $\mb N^k$ leads us to a first characterisation of $X_k$. We conclude the section by making this characterisation even more precise, in the case $k=1$.

\subsubsection*{Colourings approach}

The points of $V_k$ have a behaviour that is, in a way, similar to finite words. 
Let us explain the idea of this analogy in the case $k=1$. A way to encode a finite word is by following the insight of logic on words: it is equivalent to consider a finite word and a family of finite disjoint subsets of $\mb N$, possibly empty for some of them, which cover the initial segment $\{0,...,n-1\}$, for some integer $n \geq 1$. Labelling each of these subsets with a letter of the alphabet, this amounts to grouping together the positions of the word which correspond to the same letter. This is the definition of $c_{a}(w)$, where $a \in A$ and $w \in A^*$.
Note that, since $V_1$ is the image of $c^1$, the family $c^{1}(w)=(\wh{c_{a}(w)})_{a \in A}$ is a point of $V_1$. 
Now, fix a point $\C=(C_a)_{a \in A} \in V_1$. If we view an ultrafilter $\alpha \in \beta(\mb N)$ as a generalized position, then just like in the previous situation, we could say that $\C$ has the letter $a \in A$ at the generalized position $\alpha$ if, and only if, $\alpha \in C_a$. 
This reasoning applies for any $k \geq 1$, and for this reason, we refer to points of $\C \in V_k$ as \emph{generalized words}. \newline In the particular case where a point $\C \in V_k$ is in the subspace $X_k$, we can make this analogy with words even more precise. The following characterisation of $X_k$ is a direct consequence of the universal property of \v Cech-Stone compactification, and relates $\C$ to the existence of a certain ultrafilter in $\beta(A^{*})$. In particular, if the corresponding ultrafilter is trivial, i.e there exists a finite word $w \in A^*$ such that $\gamma = {\uparrow}\{w\}$, then for every $\a \in A^k$, $C_\a$ corresponds to the clopen associated to the set of $k$-tuples of positions $c_{\a}(w)$.

\begin{proposition}\label{prop:carcImultrafilter}
For any $k \geq 1$, $\C=(C_{\a})_{\a \in A^k}$ is in $X_k$ if, and only if, there exists an ultrafilter $\gamma \in \beta(A^{*})$ such that, for every $\a \in A^k$, and for every $Q \subseteq \mb N^k$,
\[C_\a \in \Box^{\a}(\wh{Q}) \text{ if, and only if, } L_{\Box^{\a}_Q}\in \gamma;\]
or, equivalently,
\[C_\a \in \Diam^{\a}(\wh{Q}) \text{ if, and only if, } L_{\Diam^{\a}_Q}\in \gamma.\]
\end{proposition}

\begin{proof}
By Proposition \ref{prop:imck}, a generalized word $\C=(C_{\a})_{\a \in A^k}$ is in $X_k$ if, and only if, it is in $Im(c^{k})$, that is if there exists an ultrafilter $\gamma \in \beta(A^{*})$ such that $\C =c^{k}(\gamma)$.
Now, by definition, $\C =c^{k}(\gamma)$ if, and only if, for every $\a \in A^{k}$, $C_{\a} = c_{\a}(\gamma)$.

On the one hand, by Remark \ref{remark:PUfilter} we know that
\[c_{\a}(\gamma) = \bigcap_{Q \subseteq \mb N^k \atop L_{\Box^{\a}_{Q}} \in \gamma}\wh{Q}.\]
On the other hand,
\[C_{\a} = \bigcap_{Q \subseteq \mb N^k \atop C_{\a} \subseteq \wh{Q}}\wh{Q} = \bigcap_{Q \subseteq \mb N^k \atop C_{\a} \in \Box^{\a}(\wh{Q})}\wh{Q},\]
thus for every $\a \in A^{k}$, $C_{\a} = c_{\a}(\gamma)$ if, and only if, for every $Q \subseteq \mb N^k$,
\[ C_{\a} \in \Box^{\a}\wh{Q} \text{   if, and only if,   } L_{\Box^{\a}_{Q}} \in \gamma.\]
Finally, \[ \Diam^{\a}\wh{Q}= (\Box^{\a}\wh{Q}^{c})^{c}  \text{   and   } L_{\Diam^{\a}_Q} = (L_{\Box^{\a}_{Q^{c}}})^{c} \]
thus we conclude that $\C=(C_{\a})_{\a \in A^k}$ is in $X_k$ if, and only if, there exists $\gamma \in \beta(A^{*})$ such that, for every $\a \in A^{k}$ and every $Q \subseteq \mb N^k$,
\[C_\a \in \Diam^{\a}(\wh{Q}) \text{ if, and only if, } L_{\Diam^{\a}_Q}\in \gamma.\]
\end{proof}

This motivates the following terminology: we refer to points $\C \in X_k$ as \emph{pseudofinite words}. 
This setting allows us to compute some elementary instances of pseudofinite words which are not finite.

\begin{example}
In the case where $k=1$ and $|A|=\{a,b\}$, let us consider the couple of closed subsets of $\beta(\mb N)$ \[\C = (C_a,C_b) :=(\beta(\mb N), {^{*}}\mb N).\]
We prove that this is a pseudofinite word by giving a description of an ultrafilter $\gamma$
satisfying the condition introduced in Proposition \ref{prop:carcImultrafilter}.

By Example \ref{example:filtersclosedsets}, an ultrafilter $\gamma \in \beta(A^{*})$ satisfying the condition introduced in Proposition \ref{prop:carcImultrafilter} has to be such that:

\begin{enumerate}
\item $\{P \subseteq \mb N \colon L_{\Box^{a}_{P}} \in \gamma \}  = \{\mb N\}$;
\item $\{P \subseteq \mb N \colon L_{\Box^{b}_{P}}  \in \gamma \} = Cof(\mb N).$
\end{enumerate} 

that is

\begin{enumerate}
\item $L_{\Box^{a}_{\mb N}} \in \gamma$ and, for every $S \subsetneq \mb N$, $L_{\Box^{a}_{S}} \notin \gamma$.
\item For every cofinite subset $S$ of $\mb N$, $L_{\Box^{b}_{S}} \in \gamma$ and, for every non-cofinite subset $S$ of $\mb N$, $L_{\Box^{b}_{S}} \notin \gamma$.
\end{enumerate} 

Let us reformulate these conditions.

\begin{enumerate}
\item $L_{\Box^{a}_{\mb N}}$ is equal to $A^{*}$, thus the first condition always holds. 
 We prove that the second condition is equivalent to saying that, for every $n \in \mb N$, $L_{\Diam^{a}_{\{n\}}} \in \gamma$. 
On the one hand, assume that, for every $S \subsetneq \mb N$, $L_{\Box^{a}_{S}} \notin \gamma$.
Then, in particular, for any $n \in \mb N$, $L_{\Box^{a}_{\mb N \setminus \{n\}}} \notin \gamma$, which is equivalent to $L_{\Diam^{a}_{\{n\}}} \in \gamma$.
On the other hand, assume that, for every $n \in \mb N$, $L_{\Diam^{a}_{\{n\}}} \in \gamma$. For any subset $S$ strictly contained in $\mb N$, pick $n \in S^{c}$.  We have that $L_{\Diam^{a}_{\{n\}}} \subseteq L_{\Diam^{a}_{S^c}} = (L_{\Box^{a}_{S}})^{c}$. By upset, this last language is in $\gamma$, and thus we conclude that $L_{\Box^{a}_{S}} \notin \gamma$.
\item We prove, in a similar fashion, that saying that, for every cofinite subset $S$ of $\mb N$, $L_{\Box^{b}_{S}} \in \gamma$ is equivalent to saying that, for every $n \in \mb N$, $(L_{\Diam^{b}_{\{n\}}})^{c} \in \gamma$. We also have that saying that, for every non-cofinite subset $S$ of $\mb N$, $L_{\Box^{b}_{S}} \notin \gamma$ is equivalent to saying that, for any infinite subset $S$ of $\mb N$, $L_{\Diam^{b}_{S}} \in \gamma$.
\end{enumerate}

We conclude that the condition $\gamma \in \beta(A^{*})$ has to satisfy can be rephrased as follows.

\begin{enumerate}
\item For every $n \in \mb N$, $L_{\Diam^a_{\{n\}}} \in \gamma$ 
\item For every $n \in \mb N$, $(L_{\Diam^b_{\{n\}}})^{c} \in \gamma$ and for every infinite subset $S \subseteq \mb N$, $L_{\Diam^b_{S}} \in \gamma$. 
\end{enumerate}  

Now, for any $n_1,...,n_{l_1},m_1,...,m_{l_2} \in \mb N$, where $l_1, l_2 \geq 1$, and for any finite family of infinite subsets $S_1,...,S_l$, the language 
\[ \bigcap_{i=1}^{l_1}L_{\Diam^{a}_{\{n_i\}}} \cap \bigcap_{i=1}^{l_2} (L_{\Diam^{b}_{\{m_i\}}})^{c} \cap  \bigcap_{i=1}^{l}L_{\Diam^{b}_{S_i}} \]
is non-empty, thus 
\[ {\uparrow}\{L_{\Diam^{a}_{\{n\}}},(L_{\Diam^{b}_{\{m\}}})^{c},L_{\Diam^{b}_{S}} \colon n,m \in \mb N, S\subseteq \mb N \text{ and $S$ infinite}\} \]
is a filter and by Stone's theorem it can be extended into an ultrafilter $\gamma$ which satisfies, by construction, the condition introduced in Proposition \ref{prop:carcImultrafilter}.
\end{example}

Another, and potentially more practical, characterisation of $X_k$, for any $k \geq 1$, can be made by using the terminology of finite colourings, commonly used in Ramsey theory.  We could summarize this characterisation by saying that a generalized word $\C$ is pseudofinite if, and only if, for every finite colouring, it is possible to construct an actual finite word $w \in A^*$ which is equivalent to $\C$ when we look at it from the perspective of this finite colouring.

\begin{definition}\label{def:colouring}
A \emph{finite colouring} of $\mb N^k$ is a map $q: \mb N^k \to I$, where $I$ is a finite set, or, equivalently, a finite family of pairwise disjoints subsets of $\mb N^k$, $\mc Q=(Q_{i})_{i \in I}$, such that $\bigcup_{i \in I} Q_i = \mb N^k$.
\end{definition}

First, let us define the notion of content of a word on a subset of $\mb N^k$, that is the set of $k$-tuples of letters of $w$ which occur on the given subset. This generalizes the notation $w[\i]$, where  $\i \in \mb N^k$, introduced in Definition \ref{def:semanticswords}.

\begin{definition}\label{def:contentonsubset}
For any finite word $w \in A^*$ and any subset $Q \subseteq \mb N^k$, we define the \emph{content of $w$ on $Q$} as 
\[ \langle w, Q \rangle := \{\a \in A^k \colon c_{\a}(w) \cap Q \neq \emptyset \}.\]
\end{definition}

Note that the map $\langle w, \cdot \rangle: \mc P( \mb N^k) \to \mc P(A^k)$ which sends a subset $Q$ to $\langle w, Q \rangle$ is finitely additive: for any finite family of subsets $Q_1,...,Q_n \subseteq \mb N^k$, 

\begin{align} \label{property:w[Q]finitelyadditive}
\bigcup_{i=1}^{n}\langle w, Q_i \rangle 
&= \bigcup_{i=1}^{n} \{\a \in A^k \colon c_{\a}(w) \cap Q_i \neq \emptyset    \} \notag    \\
&= \{\a \in A^k \colon c_{\a}(w) \cap \bigcup_{i=1}^n Q_i \neq \emptyset \} \notag \\
&= \langle w, \bigcup_{i=1}^{n}Q_i \rangle.
\end{align}

More generally, for any finite colouring $q : \mb N^k \to I$ of $\mb N^k$, we define the (color) profile of a finite word $w$ as 
\[\langle w,q \rangle := (\langle w, q^{-1}(i) \rangle)_{i \in I} \in \mc P(A^{k})^{I},\]
that is, the family of contents of $w$ associated to each colour.

\begin{example}
Fix $A = \{a,b\}$. In the case $k=2$, let us consider $q$, the three colours colouring $(\Delta^{<},\Delta,\Delta^{>})$ of $\mb N^2$, where
\[ \textcolor{blue}{\Delta^{<} := \{(n,m) \in \mb N^2 \colon n<m \}}, \]
\[\textcolor{qqwuqq}{\Delta^{>} := \{(n,m) \in \mb N^2 \colon n>m \}}, \]
\[\textcolor{red}{\Delta := \{(n,n) \in \mb N^2 \colon n \in \mb N \}}\]
and the finite word $w=ababb$.

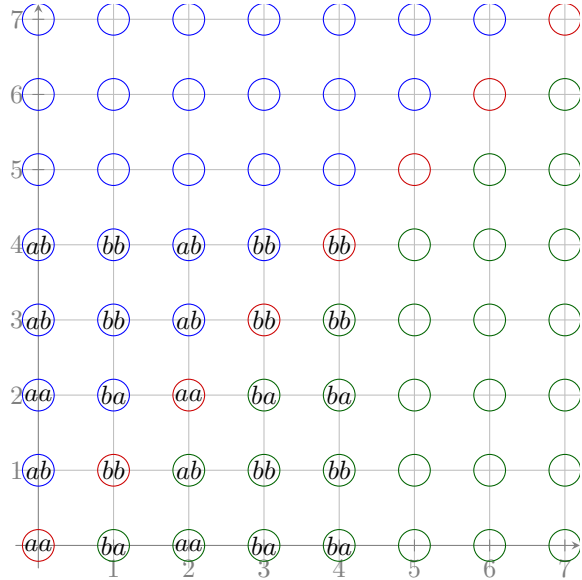
\begin{center}
\begin{tikzpicture}[line cap=round,line join=round,>=triangle 45,x=1cm,y=1cm]
\begin{axis}[gray,
x=1cm,y=1cm,
axis lines=middle,
ymajorgrids=true,
xmajorgrids=true,
xmin=-0.3,
xmax=7.2,
ymin=-0.3,
ymax=7.2,
xtick={0,-2,...,7},
ytick={0,-1,...,7},]
\clip(-3.3372307032278825,-2.6240526184601527) rectangle (15.249424098624392,8.360698378953552);
\begin{scriptsize}
\draw [color=ccqqqq] (0,0) circle (6pt);
\draw[color=black] (0,0) node {$aa$};
\draw [color=ccqqqq] (1,1) circle (6pt);
\draw[color=black] (1,1) node {$bb$};
\draw [color=ccqqqq] (2,2) circle (6pt);
\draw[color=black] (2,2) node {$aa$};
\draw [color=ccqqqq] (3,3) circle (6pt);
\draw[color=black] (3,3) node {$bb$};
\draw [color=ccqqqq] (4,4) circle (6pt);
\draw[color=black] (4,4) node {$bb$};
\draw [color=ccqqqq] (5,5) circle (6pt);
\draw [color=ccqqqq] (6,6) circle (6pt);
\draw [color=ccqqqq] (7,7) circle (6pt);
\draw [color=qqqqff] (2,7) circle (6pt);
\draw [color=qqqqff] (3,7) circle (6pt);
\draw [color=qqqqff] (1,7) circle (6pt);
\draw [color=qqqqff] (1,6) circle (6pt);
\draw [color=qqqqff] (1,5) circle (6pt);
\draw [color=qqqqff] (1,4) circle (6pt);
\draw[color=black] (1,4) node {$bb$};
\draw [color=qqqqff] (2,4) circle (6pt);
\draw[color=black] (2,4) node {$ab$};
\draw [color=qqqqff] (2,5) circle (6pt);
\draw [color=qqqqff] (2,6) circle (6pt);
\draw [color=qqqqff] (3,6) circle (6pt);
\draw [color=qqqqff] (4,7) circle (6pt);
\draw [color=qqqqff] (5,7) circle (6pt);
\draw [color=qqqqff] (6,7) circle (6pt);
\draw [color=qqqqff] (5,6) circle (6pt);
\draw [color=qqqqff] (4,6) circle (6pt);
\draw [color=qqqqff] (4,5) circle (6pt);
\draw [color=qqqqff] (3,5) circle (6pt);
\draw [color=qqqqff] (3,4) circle (6pt);
\draw[color=black] (3,4) node {$bb$};
\draw [color=qqqqff] (1,3) circle (6pt);
\draw[color=black] (1,3) node {$bb$};
\draw [color=qqqqff] (1,2) circle (6pt);
\draw[color=black] (1,2) node {$ba$};
\draw [color=qqqqff] (2,3) circle (6pt);
\draw[color=black] (2,3) node {$ab$};
\draw [color=qqqqff] (0,7) circle (6pt);
\draw [color=qqqqff] (0,6) circle (6pt);
\draw [color=qqqqff] (0,5) circle (6pt);
\draw [color=qqqqff] (0,1) circle (6pt);
\draw[color=black] (0,1) node {$ab$};
\draw [color=qqqqff] (0,2) circle (6pt);
\draw[color=black] (0,2) node {$aa$};
\draw [color=qqqqff] (0,3) circle (6pt);
\draw[color=black] (0,3) node {$ab$};
\draw [color=qqqqff] (0,4) circle (6pt);
\draw[color=black] (0,4) node {$ab$};
\draw [color=qqwuqq] (7,6) circle (6pt);
\draw [color=qqwuqq] (7,5) circle (6pt);
\draw [color=qqwuqq] (7,4) circle (6pt);
\draw [color=qqwuqq] (7,3) circle (6pt);
\draw [color=qqwuqq] (7,2) circle (6pt);
\draw [color=qqwuqq] (7,1) circle (6pt);
\draw [color=qqwuqq] (6,4) circle (6pt);
\draw [color=qqwuqq] (6,5) circle (6pt);
\draw [color=qqwuqq] (5,4) circle (6pt);
\draw [color=qqwuqq] (5,3) circle (6pt);
\draw [color=qqwuqq] (4,3) circle (6pt);
\draw[color=black] (4,3) node {$bb$};
\draw [color=qqwuqq] (4,2) circle (6pt);
\draw[color=black] (4,2) node {$ba$};
\draw [color=qqwuqq] (3,2) circle (6pt);
\draw[color=black] (3,2) node {$ba$};
\draw [color=qqwuqq] (3,1) circle (6pt);
\draw[color=black] (3,1) node {$bb$};
\draw [color=qqwuqq] (2,1) circle (6pt);
\draw[color=black] (2,1) node {$ab$};
\draw [color=qqwuqq] (4,1) circle (6pt);
\draw[color=black] (4,1) node {$bb$};
\draw [color=qqwuqq] (5,1) circle (6pt);
\draw [color=qqwuqq] (5,2) circle (6pt);
\draw [color=qqwuqq] (6,3) circle (6pt);
\draw [color=qqwuqq] (6,2) circle (6pt);
\draw [color=qqwuqq] (6,1) circle (6pt);
\draw [color=qqwuqq] (3,0) circle (6pt);
\draw[color=black] (3,0) node {$ba$};
\draw [color=qqwuqq] (4,0) circle (6pt);
\draw[color=black] (4,0) node {$ba$};
\draw [color=qqwuqq] (5,0) circle (6pt);
\draw [color=qqwuqq] (6,0) circle (6pt);
\draw [color=qqwuqq] (7,0) circle (6pt);
\draw [color=qqwuqq] (2,0) circle (6pt);
\draw[color=black] (2,0) node {$aa$};
\draw [color=qqwuqq] (1,0) circle (6pt);
\draw[color=black] (1,0) node {$ba$};
\end{scriptsize}
\end{axis}
\end{tikzpicture}
\end{center}

The profile of $w$ for the colouring $q$ is
\[ \langle w,q \rangle = (\langle w, \Delta^{<} \rangle, \langle w,\Delta \rangle, \langle w, \Delta^{>} \rangle) = (A^{2}, \{aa,bb\},A^{2}).\]
\end{example}

In our framework, a natural idea is to extend this notion of (colour) profile to generalized words. 

\begin{definition}\label{def:colour}
For any subset $Q$ of $\mb N^k$, and any generalized word $\C=(C_a)_{\a \in A^k} \in V_k$, we define \emph{the content of $\C$ on $Q$} as
\[\langle \C ,Q \rangle := \{\a \in A^k \colon  C_\a \cap \wh{Q} \neq \emptyset  \}.\]
The \emph{colour profile} of a generalized word $\C=(C_a)_{\a \in A^k} \in V_k$ on a finite colouring $q: \mb N^k \to I$ is the map \[ \langle \C,q \rangle : I \to \mc P(A^k)\]
which sends any $i \in I$ to 
\[  \langle \C,q \rangle (i) := \langle \C, q^{-1}(i) \rangle.\]
\end{definition}

In particular, if the point of $V_k$ we consider is of the form $(\wh{c_\a(w)})_{\a \in A^k}$ for some finite word $w \in A^{*}$, then for any finite colouring $q: \mb N^k \to I$, we have 
\[ \langle (\wh{c_{\a}(w)})_{\a \in A^k},q \rangle = \langle w, q \rangle,\]
which shows that the profile of a generalized word can be seen as an extension of the notion of profile of a finite word.  
Also, notice that for any $\C \in V_k$, the map $\langle \C, \cdot \rangle : \mc P(\mb N^k) \to \mc P(A^k)$ which sends a subset $Q$ to $\langle \C,Q \rangle$ is finitely additive: for any finite family of subsets $Q_1,...,Q_n \subseteq \mb N^k$,

\begin{align*}
\bigcup_{i=1}^{n} \langle \C, Q_i \rangle
&= \bigcup_{i=1}^{n} \{\a \in A^k \colon C_\a \cap \wh{Q_i} \neq \emptyset    \}   \\
&= \{\a \in A^k \colon C_\a \cap \bigcup_{i=1}^n \wh{Q_i} \neq \emptyset \}\\
&= \{\a \in A^k \colon C_\a \cap \wh{\bigcup_{i=1}^n Q_i} \neq \emptyset \}\\
&= \langle \C,  \bigcup_{i=1}^{n}Q_i \rangle.
\end{align*}

We are now going to use the family of all finite colourings in order to provide a different basis for the space $V_k$. 

\begin{lemma}\label{lemma:Imcolouring}
For any finite colouring $q: \mb N^k \to I$ of $\mb N^k$, we consider the map 
\[ \langle \cdot ,q \rangle: V_k \to \mc P(A^k)^{I} \] which sends any $\C \in V_k$ to $\langle \C,q \rangle $.
We also consider the family of all preimages, for all of these maps
\[ \mc C:= \{\langle \cdot ,q \rangle ^{-1}(\B) \colon  q: \mb N^k \to I, \text{ where $I$ is a finite set, and }   \B \in \mc P(A^k)^I \}.\] 
The following statements hold. \newline (1): The inverse image of any point $\B=(B_i)_{i \in I} \in \mc P(A^k)^{I}$ under $\langle \cdot ,q \rangle $ is clopen in $V_k$. In particular, for any finite colouring $q: \mb N^k \to I$ of $\mb N^k$, the map $\langle \cdot, q \rangle$ is continuous when $\mc P(A^{k})^{I}$ is equipped with the discrete topology. \newline  (2): 
Any intersection of two elements in $\mc C$ can be written as a finite union of elements in $\mc C$. In particular, $\mc C$ is a basis for the topology on $V_k$.  
\end{lemma}

\begin{proof}
(1): Fix $q: \mb N^k \to I$ a finite colouring of $\mb N^k$, and a family of subsets $\B=(B_i)_{i \in I} \in \mc P(A^k)^{I}$.
Recall that, for any $Q \subseteq \mb N^k$, the subsets of the form \[\Diam Q = \{C \in \mc V(\beta(\mb N^k))  \colon C \cap \wh{Q} \neq \emptyset \} \] are clopen in $\mc V(\beta(\mb N^k))$.
We can express $\langle \cdot ,q \rangle^{-1} (\B)$ as a finite Boolean combination of these clopen subsets:
\begin{align*}
\langle \cdot ,q \rangle^{-1} (\B) &= \{ \C = (C_\a)_{\a \in A^k} \in V_k \colon \langle \C, q \rangle = \B \} \\
&=  \{ \C = (C_\a)_{\a \in A^k} \in V_k \colon \forall i \in I, \langle \C, q \rangle(i) = B_i\}\\
&=  \{ \C = (C_\a)_{\a \in A^k} \in V_k \colon \forall i \in I, \forall \a \in A^k, (
C_\a \cap \wh{q^{-1}(i)} \neq \emptyset \Longleftrightarrow \a \in B_i )\} \\
&= \bigcap_{i \in I} (\bigcap_{\a \in B_i} p_{\a}^{-1}(\Diam \wh{q^{-1}(i)}) \cap \bigcap_{\a \notin B_i}p_{\a}^{-1}(\Diam \wh{q^{-1}(i))^{c}}),
\end{align*}
where, for every $\a \in A^k$, $p_{\a}: V_k \to \mc V(\beta(\mb N^k))$ sends any $\C=(C_{\a})_{\a \in A^k}$ to $C_\a$.
Since $\mc P(A^k)^I$ is equipped with the discrete topology, this proves that the map $\langle \cdot, q \rangle$ is continuous.\newline (2): Fix two finite colourings $\mc Q = (Q_1,...,Q_l)$ and $\mc Q' = (Q'_1,...,Q'_n)$ of $\mb N^k$, where $l,n \geq 1$. 
Fix $\B =(B_1,...,B_l) \in \mc P(A^{k})^{l}$ and $\B'=(B'_1,...,B'_n) \in \mc P(A^{k})^{n}$. 
First, we define a finite colouring of $\mb N^k$ which refines both $\mc Q$ and $\mc Q'$:
$\mc R =(R_{i,j})_{1 \leq i \leq l \atop 1 \leq j \leq n}$, the finite colouring of $\mb N^k$, such that, for every $(i,j) \in \{1,...,l\} \times \{1,...,n\}$, \[R_{i,j}:= Q_i \cap Q'_j.\]
Finally, we define $\mc D_{\B,\B'} \subseteq \mc P(A^{k})^{l.n}$ as follows: 
$\D = (D_{i,j})_{1 \leq i \leq l \atop 1 \leq j \leq n}$ is in $\mc D_{\B,\B'}$ if, and only if, 
for every $i \in \{1,...,l\}$, 
\[ \bigcup \{ D_{u,v} \colon (u,v) \in \{1,...,l\}\times \{1,...,n\} \text{  and  } R_{u,v} \subseteq Q_i \}  = B_i \]
and for every $j \in \{1,...,n\}$,
\[ \bigcup \{ D_{u,v} \colon (u,v) \in \{1,...,l\}\times \{1,...,n\} \text{  and  } R_{u,v} \subseteq Q'_j \}  = B'_j.\]
We now prove that \[ \langle \cdot , \mc Q \rangle^{-1}(\B) \cap  \langle \cdot , \mc Q' \rangle^{-1}(\B') = \bigcup_{\D \in \mc D_{\B,\B'}} \langle \cdot , \mc R \rangle^{-1}(\D).\] For the left-to-right inclusion, fix $\C \in \langle \cdot , \mc Q \rangle^{-1}(\B) \cap  \langle \cdot , \mc Q' \rangle^{-1}(\B')$. We define $\D_{\C} \in \mc P(A^{k})^{l.n}$ as follows: set, for any $(u,v) \in \{1,...,l\} \times \{1,...,n\}$,
\[(\D_{\C})_{u,v} := \langle \C, R_{u,v} \rangle   .\]
By construction, it is clear that $\C$ belongs to $\langle \cdot , \R  \rangle ^{-1}(\D_{\C})$. All we have left to prove in order to conclude is that $\D_{\C}$ is in $\mc D_{\B,\B'}$.
By definition of $\D_{\C}$ we have that, for any $i \in \{1,...,l\}$,
\[\bigcup_{(u,v) \in \{1,...,l\} \times \{1,...,n\} \atop R_{u,v} \subseteq Q_i } (D_{\C})_{u,v}  
 = \bigcup_{(u,v) \in \{1,...,l\} \times \{1,...,n\} \atop R_{u,v} \subseteq Q_i } \langle \C, R_{u,v} \rangle.\]
Now, since the map $\langle \cdot , \R \rangle : \mc P(\mb N^k) \to \mc P(A^{k})^{l.n}$ is finitely additive, this is also equal to
\[\langle \C, \bigcup_{(u,v) \in \{1,...,l\} \times \{1,...,n\} \atop R_{u,v} \subseteq Q_i } R_{u,v} \rangle,\]
that is, $\langle \C, Q_i \rangle$ and since $\C$ is in $\langle \cdot,\mc Q \rangle^{-1}(\B)$ we finally obtain that
\[\bigcup_{(u,v) \in \{1,...,l\} \times \{1,...,n\} \atop R_{u,v} \subseteq Q_i } (D_{\C})_{u,v} = B_i.\]
We prove in the exact same way that, for any $j \in \{1,...,n\}$,
\[\bigcup_{(u,v) \in \{1,...,l\} \times \{1,...,n\} \atop R_{u,v} \subseteq Q_i } (D_{\C})_{u,v}= B'_j,\]
which ends to prove that $\D_{\C} \in \mc D_{\B,\B'}$ and allows us to conclude. \newline For 
the left-to-right inclusion, fix $\D_{\C} \in \mc D_{\B,\B'}$.
For any $\C \in \langle \cdot, \mc R \rangle^{-1}(\D)$, we have that, for every $i \in \{1,...,l\}$,
\[ \langle \C, Q_i \rangle  = \langle \C, \bigcup_{(u,v) \in \{1,...,l\} \times \{1,...,n\} \atop R_{u,v \subseteq Q_i}}R_{u,v} \rangle.\]
Now, since the map $\langle \cdot , \R \rangle : \mc P(\mb N^k) \to \mc P(A^{k})^{l.n}$ is finitely additive, this is also equal to
\[ \bigcup_{(u,v) \in \{1,...,l\} \times \{1,...,n\} \atop R_{u,v \subseteq Q_i}} \langle \C , R_{u,v} \rangle \]
and since $\D \in \mc D_{\B,\B'}$, this is equal to
\[\bigcup_{(u,v) \in \{1,...,l\} \times \{1,...,n\} \atop R_{u,v \subseteq Q_i}} D_{u,v}\]
which allows us to conclude that
$\langle \C, Q_i \rangle = B_i.$ The exact same reasoning can be conducted to prove that $\langle \cdot,\R \rangle^{-1}(\D) \subseteq \langle \cdot, \Q' \rangle^{-1}(\B')$, which allows us to conclude. \newline Finally, we prove that $\mc C$ is a basis for the topology on $V_k$. By (2), all we have left to prove in order to do so is that every $\C \in V_k$ is contained in an element of $\mc C$. Considering the one element colouring $\mc Q_{\C}:= \{\mb N^k\}$ of $\mb N^k$ and setting $\B_{\C}:= \{\a \in A^{k} \colon C_{\a} \neq \emptyset \}$, we have that
\[ \langle \C , \mc Q \rangle  = \{ \a \in A^{k} \colon C_{\a} \cap \wh{\mb N^k} \neq \emptyset \} = \B_{\C},\]
which allows us to conclude that $\mc C$ is a basis for the topology on $V_k$.
\end{proof}

We use this basis for the topology on $V_k$ in order to characterise $X_k$: the pseudofinite words $\C$ in $V_k$ are exactly the points such that, for each finite colouring $q$, we can construct a concrete finite word $w_q \in A^{*}$ which has the same profile than $\C$ on $q$.

\begin{proposition}\label{prop:caraImccolouring}
A generalized word $\C \in V_k$ is pseudofinite if, and only if, for every finite colouring $q$ of $\mb N^k$, there exists a finite word $w_{q} \in A^*$ such that the profiles of $\C$ and $w_q$ on $q$ coincide.
In particular, for any subset $Q$ of $\mb N^k$ which is saturated with respect to $q$, we have, for every $\a \in A^k$, that
\[C_\a \cap \wh{Q} \neq \emptyset \text{  if, and only if,  } c_{\a}(w_q) \cap Q \neq \emptyset.\]
\end{proposition}

\begin{proof} 
A generalized word $\C \in V_k$ is pseudofinite if, and only if, it is in $X_k$.
Recall that a $X_k$ is the closure of the image of $A^*$ under the map $c^k : A^* \to V_k$, which sends a finite word $w \in A^*$ to $(\wh{c_{\a}(w)})_{\a \in A^k}$. We proved in Lemma \ref{lemma:Imcolouring} (2), that the family 
\[ \mc C= \{\langle \cdot ,q \rangle ^{-1}(\B) \colon \B \in \mc P(A^k)^I \text{  and  } q: \mb N^k \to I, \text{ where $I$ is a finite set} \}\] forms a basis for the topology on $V_k$.
Therefore, the characterization of topological closure by a basis provides the following characterisation of $X_k$:
a generalized word $\C \in V_k$ is in $X_k$ if, and only if, 
for every finite colouring $q : \mb N^k \to I$ of $\mb N^k$, and every $\B \in \mc P(A^{k})^{I}$ such that $\langle \C,q \rangle = \B$, we have \[\langle \cdot ,q \rangle ^{-1}(\B) \cap c^{k}(A^{*}) \neq \emptyset.\]
Note that this last condition is equivalent to saying that there exists a finite word $w_q$ such that \[\langle (\wh{c_{\a}(w_q)})_{\a \in A^k},q) = \B.\]
We previously observed that the profile of $(\wh{c_{\a}(w_q)})_{\a \in A^k}$ on $q$ is the profile of $w_q$ on $q$. We conclude that $\C \in V_k$ is pseudofinite if, and only if, there exists a finite word $w_q$ such that $\langle \C,q \rangle = \langle w_{q},q \rangle$. \newline The other statement is a direct consequence of the fact that, for any finite word $w\in A^*$, the map $\langle w, \cdot \rangle : \mc P(\mb N^k) \to \mc P(A^k)$, which sends a subset $Q$ to $\langle w,Q \rangle$, and for any $\C \in V_k$, the map $\langle \C, \cdot \rangle : \mc P(\mb N^k) \to \mc P(A^k)$ which sends a subset $Q$ to $\{\a \in A^k \colon C_a \cap \wh{Q} \neq \emptyset \}$ are finitely additive. 
\end{proof}

\begin{remark} \label{remark:contentfinite}
In particular, for any pseudofinite word $\C \in X_k$, let us consider a finite colouring $q: \mb N^k \to I$ such that one of the colours corresponds to a singleton, that is
\[q^{-1}(i)= \{\p\}\] for some $i \in I$ and $\p = (p_1,...,p_k) \in \mb N^k$.
In that case, we observe that any word $w$ satisfying the condition from Proposition \ref{prop:caraImccolouring} is necessarily such that 
\[|w| > \max \{p_j \colon j \in \{1,...,k\} \} \]
and such that, for any $\a \in A^{k}$,  
\[ \p \in Cont(C_{\a}) \text{  if, and only if,   }w[\p]=\a.\]
Indeed, for any $\a \in A^{k}$,
\begin{align*}
\p \in Cont(C_\a) 
&\Longleftrightarrow \p \in C_{\a} \cap \mb N^{k} \\
&\Longleftrightarrow \a \in \{\b \in A^{k} \colon C_{\b} \cap \{\p\} \neq \emptyset \} \\
&\Longleftrightarrow \a \in \langle w,q \rangle(i) \text{     by Proposition \ref{prop:caraImccolouring}} \\
&\Longleftrightarrow  w[\p]=\a.
\end{align*}

This remark will come handy in the proof of Lemma \ref{lemma:CNpartitioncontent}.
\end{remark}

\begin{example}
In the case $k=1$ and $|A|=\{a,b\}$ let us consider the family of closed subset of $\beta(\mb N)$ 
\[\C :=(\beta(\mb N),\emptyset).\] Intuitively, this should be a pseudofinite word, that we could see as a generalization of the profinite word $a^{\omega}$. In practice, we can apply Proposition \ref{prop:caraImccolouring}: for any finite colouring $q: \mb N \to I$, where $I$ is a finite set, we set, for every $i \in I$, $n_i:= \min(q^{-1}(i))$ and 
\[N:=\max_{i \in I}n_i.\]  
We consider the word $w_{q}:=a^N$. 
This allows us to prove that this family of closed subsets is a pseudofinite word.
\end{example}

\subsubsection*{Explicit characterisation of $X_1$}

In the case where $k=1$, it is not too difficult to directly simplify Proposition \ref{prop:caraImccolouring} into a condition that does not require us to look at every finite colouring of $\mb N^k$. In order to do so, we start by proving a necessary condition that holds for every $\C \in X_k$. The intuition is the following. Pseudofinite words share similarities with finite words, but at the generalized level of ultrafilters. In particular, if $\C=(C_\a)_{\a \in A^{k}} \in X_k$, then for any $\a \in A^k$, the content of $C_\a$, introduced in Definition \ref{def:contentclosedset}, is a subset of $\mb N^k$. It should be possible to view the elements of this subset as $k$-tuples of positions of a concrete word, with a length that is possibly infinite. We formalize this intuition here-below.

\begin{lemma}\label{lemma:CNpartitioncontent}
Fix $\C=(C_{\a})_{\a \in A^{k}} \in X_k$ a pseudofinite  word. Then, the following statements hold. \newline
(1): For any $\p=(p_1,...,p_k) \in \mb N^k$ and any $\a=(a_1,...,a_k) \in A^k$, we have that $\p \in Cont(C_{\a})$ if, and only if, for every $j \in \{1,...,k\}$, $(p_j,...,p_j) \in Cont(C_{a_{j},...,a_{j}})$. \newline
(2): For every $a \in A$, consider the subset of $\mb N$
\[C_{a}^{\mb N}:= \pi_1(Cont(C_{a,...,a})),\] 
where $\pi_1: \mb N^k \to \mb N$ is the canonical projections on the first coordinate. Then $(C^{\mb N}_a)_{a \in A}$ is a finite colouring of a downset of $\mb N$, and, for every $\a \in A^k$, 
\[Cont(C_\a) = \prod_{j=1}^{k} C_{a_{j}}^{\mb N}.\]
\end{lemma}

\begin{proof}
(1): Fix $\p=(p_1,...,p_k) \in \mb N^k$, and $\a=(a_1,...,a_k) \in A^{k}$. To keep the notations concise, we set, for every $j \in \{1,...,k\}$,\[\bar{p^{j}}:=(p_j,...,p_j) \text{  and  } \bar{a^{j}}:=(a_j,...,a_j).\]
We consider the following colouring of $\mb N^k$ into $k+2$ colours 
\[\mc Q_{\p}:= (\{\p\},\{\bar{p^1}\},...,\{\bar{p^k}\}, \mb N^k \setminus \{\p,\bar{p^1},...,\bar{p^k}\}).\] Since $\C$ is in $X_k$, by Proposition \ref{prop:caraImccolouring}, we can consider a finite word $w \in A^{*}$ which has the same profile than $\C$ for the colouring $\mc Q_{\p}$. In particular, $|w| > \max \{p_j \colon j \in \{1,...,k\} \}$.
We now prove the desired equivalence by using Remark \ref{remark:contentfinite},
\begin{align*}
\p \in Cont(C_{\a}) 
& \Longleftrightarrow w[\p]=\a  \\
& \Longleftrightarrow \forall j \in \{1,...,k\}, w_{p_{j}}=a_j \\
& \Longleftrightarrow \forall j \in \{1,...,k\}, w[\bar{p^{j}}]=\bar{a^{j}} \\
& \Longleftrightarrow \bar{p^{j}} \in Cont(C_{\bar{a^{j}}}).
\end{align*}
(2): First, we prove that for any $a,b \in A$ distincts, $C_a^{\mb N} \cap C_b^{\mb N}$ is empty.
Let us assume that $C_a^{\mb N} \cap C_b^{\mb N}$ is non-empty.
Pick an element $l \in C_a^{\mb N} \cap C_b^{\mb N}$, and then pick $\p \in Cont(C_{a,...,a})$ and $\m \in Cont(C_{b,...,b})$ such that $p_1=m_1=l$. 
We now consider the colouring of $\mb N^k$ into three colors 
\[\mc Q_{\p,\m}:= (\{\p\},\{\m\}, \mb N^k \setminus \{\p,\m\}).\]
Since $\C$ is in $X_k$, by Proposition \ref{prop:caraImccolouring}, we can consider a finite word $w \in A^{*}$ which has the same profile than $\C$ for the colouring $\mc Q_{\p,\m}$.
In particular, by Remark \ref{remark:contentfinite}, since $\p \in Cont(C_{a,...,a})$ and $\m \in Cont(C_{b,...,b})$, we have that
\[w[\p]=(a,...,a) \text{  and  } w[\m]=(b,...,b).\]
Now, since $p_1=m_1$, we have that $a=b$, which allows us to conclude. \newline 
Finally, we prove that $\bigcup_{a \in A}C^{\mb N}_{a}$ is a downset of $\mb N$.
Fix $a \in A$, $l \in C_{a}^{\mb N}$ and consider some element $n < l$.
We prove that there exists $b \in A$ such that $n \in C_{b}^{\mb N}$.
We use the notation $\n := (n,...,n) \in \mb N^k$.
Picking an element $\p \in C_{a,...,a}$ with $p_1=l$, we consider the colouring of $\mb N^k$ into three colors  
\[ \mc Q_{\p,\n} := (\{\p\},\{\n\}, \mb N^k \setminus \{\p,\n\}).\]
Since $\C$ is in $X_k$, by Proposition \ref{prop:caraImccolouring}, we can consider a finite word $w \in A^{*}$ which has the same profile than $\C$ for the colouring $\mc Q_{\p,\n}$.
In particular, by Remark \ref{remark:contentfinite}, $|w|>n$ and we can set $b := w_n$. We have that $\n \in C_{b,...,b}$, and thus $n \in C^{\mb N}_b$, which allows us to conclude.
\end{proof}

This condition is actually sufficient to characterise all pseudofinite words in the case $k=1$.
\begin{proposition}[Explicit description of $X_1$] \label{prop:caraImc1}
A generalized word $\C=(C_a)_{a \in A} \in V_1$ is pseudofinite if, and only if, $(Cont(C_a))_{a \in A}$ is a finite colouring of a downset of $\mb N$. 
\end{proposition}

\begin{proof}
The left-to-right implication is exactly Lemma \ref{lemma:CNpartitioncontent}, with $k=1$. 
For the right-to-left implication, let us consider a generalized word $\C=(C_a)_{a \in A} \in V_1$ such that $(Cont(C_a))_{a \in A}$ is a finite colouring of a downset of $\mb N$. In order to conclude, we prove that $\C$ satisfies the condition introduced in Proposition \ref{prop:caraImccolouring}. Fix a finite colouring $q: \mb N \to I$ of $\mb N$. We are going to construct a word $w_q$ such that the profiles of $\C$ and $w_q$ coincide for the colouring $q$.
First, in the case where $\bigcup_{a \in A} Cont(C_a)$ is finite, the word $w_q$ such that, for every $a \in A$, \[c_{a}(w_q) =  Cont(C_a)\] is finite, and has, by construction, the same profile than $\C$ on $q$.
Now, we treat the case where $\bigcup_{a \in A} Cont(C_a)$ is equal to $\mb N$: we need to make sure that the word $w_q$ we construct is long enough. 
Pick $n \in \mb N$ such that any colour occurring finitely many times does not occur after $n$, that is, \[n > max \bigcup_{i \in I \atop q^{-1}(i) \text{ finite}}q^{-1}(i).\]
Now, for every colour $i \in I$ that occurs, i.e $q^{-1}(i)$ is non-empty, and for every $a\in A$ such that $Cont(C_a) \cap q^{-1}(i)$ is non-empty, pick $m_{i,a}$ in that set, and then pick \[m> max\{m_{i,a} \colon i \in I \text{ and } a \in A \text{ such that } Cont(C_a) \cap q^{-1}(i) \neq \emptyset \}.\] Finally, set \[l:=max(m,n).\]
We now define $w_q$ as the word of length $l$ that has the letter $a$ at the position $p<l$ if, and only if, $p \in Cont(C_a)$.
Finally, we check that the word $w_q$ we constructed has the same profile than $\C$ for the colouring $q$.
Since $l > n$, it is clear that $w_q$ contains all positions corresponding to a finite colour.
For an infinite colour, we make a case distinction.
If $a \in A$ is such that $Cont(C_a) \cap q^{-1}(i)$ is non-empty, then since $l > m_{i,a}$, $w_q$ contains a position $m_{i,a} \in Cont(C_a) \cap q^{-1}(i)$ such that the associated letter is the letter $a$. Otherwise, if $a \in A$ is such that $Cont(C_a) \cap q^{-1}(i)$ is empty, then for any position $p < l$ in $w_q$, $p \in Cont(C_b)$ necessarily implies that $b \neq a$. Therefore, the profiles of $\C$ and $w_q$ coincide on $q$, and we conclude that $\C \in X_1$.
\end{proof}

\section{Ultrafilter equations for $\mc B \Sigma_1[\mc N_1^{u}]$} \label{sect:ultraequa}

An explicit basis of ultrafilter equations for $\mc B \Sigma_1[\mc N_{0},\mc N^{u}_1]$, the fragment obtained out of $\mc B \Sigma_1[\mc N_1^{u}]$ by adding nullary numerical predicates, is already available in \cite{GKP}. However, the reasoning conducted there does not directly allow for a generalization to $\mc B \Sigma_1[\mc N_{0},\mc N^{u}_k]$, for any $k \geq 2$. 
An extra step has to be performed in order to understand the general case. One reason is that the proof in \cite{GKP} does not rely on purely topology and requires several combinatorial arguments.
Our perspective reduces the combinatorics that is involved in the proofs of soudness and completness in \cite{GKP} to a bare minimum. The main ingredient we use from the previous chapter is the approach introduced in Section \ref{section:colouring}, which involves finite colourings of $\mb N$. This allows for a reformulation of the ultrafilter equations from \cite{GKP} in terms of the existence of a finite colouring of $\mb N$ which satisfies certain properties. 

\textbf{Outline of the section} 
In Section \ref{section:complementBk}, we give a different presentation of the Boolean algebra $\mc B_k$, for any $k \geq 1$, taking inspiration from the characterisation of the points of the dual space $X_k$ we provided in Proposition \ref{prop:caraImccolouring}. In Section \ref{section:family}, we describe a general family of ultrafilter equations, which will encompass every equation needed in order to describe $\mc B_1$. We then give a reformulation of this general family of equations in terms of a condition relative to a finite colourings of $\mb N$. This will simplify the reasoning which will follow, and in Section \ref{section:equasig1} we use this reformulation in order to show soundness and completeness for the equations we introduced to describe $\mc B_1$.

\subsection{An alternative presentation for $\mc B_{k}$}\label{section:complementBk}

In this section, we give a different presentation of the Boolean algebras $\mc B_k$, for any $k \geq 1$, taking inspiration from the characterisation of the points of the dual space $X_k$ we provided in Proposition \ref{prop:caraImccolouring}. This presentation, in terms of languages associated to finite colourings of $\mb N^k$, will greatly simplify the proofs of soundness and completeness which will follow, and allow for an enlightening reformulation of the ultrafilter equations we will consider. \newline
Let us explain this setting in the case where $k=1$. Let us look again at the languages $L_{\Diam^{a}_Q}$, where $a \in A$ and $Q \subseteq \mb N$, introduced at the very beginning of Section \ref{section:setting}. Instead of fixing a subset $Q$ of $\mb N$, a letter $a \in A$, and considering the set of all words such that there exists a position $i < |w|$ such that $w_i=a$, we could rather fix a finite colouring of $\mb N$, and consider the words such that the content of $w$ on each colour is exactly a given subset of $A$. The languages that we obtain this way are related to the notion of profile that we introduced in Definition \ref{def:colour}, and we prove that they allow for an alternative description of the Boolean algebra $\mc B_1$ and $\mc B_{0,1}$. This idea can be generalized for any $k \geq 1$, and this motivates the introduction of the languages that we define here-below.

In the rest of the chapter, for any $k \geq 1$, to mention a finite colouring of $\mb N^k$ with $\ell \geq 1$ colours, we use the notation $\mc Q = (Q_1,...,Q_{\ell})$.
 
\begin{definition}\label{def:KQB}
For any $Q \subseteq \mb N^k$ and any $B \subseteq A^{k}$, we consider the language $K_{Q,B}$ of all words having content $B$ on $Q$,
\[K_{Q,B} := \{w \in A^{*} \colon \langle w, Q \rangle =B \}.\]
More generally, for any $k, \ell\geq 1$, any finite colouring $\mc Q$ of $\mb N^k$ with $\ell$ colours, and for any family $\B = (B_1,...,B_\ell)$ of $\ell$ subsets of $A^{k}$, we consider the language of the words having content on $Q_j$ equal to $B_j$, for every $j \in \{1,...,\ell\}$,
\[K_{\mc Q,\B} := \bigcap_{j=1}^{l} K_{Q_{j},B_{j}}.\]
\end{definition}

Observe that these languages can be seen as the equivalence classes for a certain equivalence relation on $A^{*}$.
Indeed, for any $Q \subseteq \mb N^k$, we set
\[\sim_Q := \{(w_1,w_2) \in (A^{*})^{2} \colon \langle w_1,Q \rangle= \langle w_2,Q \rangle\}.\]  
Now, for any finite word $w_1 \in A^{*}$,  setting $B :=\langle w_1 ,Q \rangle$, the equivalence class which contains $w_1$ is
\[ [w_1]_{\sim_Q} = \{w_2 \in A^{*} \colon \langle w_1,Q \rangle= \langle w_2 , Q \rangle \} = \{w_2 \in A^{*} \colon \langle w_1 ,Q \rangle=B \} = K_{Q,B}.\]
Note that this equivalence relation is finitely indexed, since $\mc P(A^k)$ is finite.

More generally, for any $k,\ell\geq 1$, and any finite colouring of $\mb N^k$ with $\ell$ colours $\mc Q$, we set
\[\sim_{\mc Q} := \bigcap_{j=1}^\ell \sim_{Q_j},\]
and an equivalence class for this relation corresponds to a language of the form $K_{\mc Q,\B}$, for some $\B \in \mc P(A^k)^{\ell}$. This equivalence relation is also finitely indexed, since $\mc P(A^{k})^{\ell}$ is finite.
These languages allow for a reformulation of the generators of the Boolean algebra $\mc B_k$, for any $k \geq 1$, which will greatly simplify our considerations in the upcoming sections. 

\begin{proposition} \label{prop:caraB_kKQB}
For any $k \geq 1$, the Boolean algebra $\mc B_k$ is generated by the languages $K_{\mc Q,\B}$, where $\mc Q$  ranges over finite colourings of $\mb N^k$ with $\ell \geq 1$ colours and $\B$ ranges over $\mc P(A^{k})^{\ell}$. 
\end{proposition} 

\begin{proof} We know, by Proposition \ref{preuve1} that $\mc B_k$ is generated by the languages $L_{\Diam^{\a}_Q}$ where $Q$ ranges over subsets of $\mb N^k$ and $\a$ ranges over $A^k$.
First, we prove that these languages can be expressed as a Boolean combination of languages of the form $K_{\mc Q,\B}$.
Let us consider the two colours colouring $(Q,Q^c)$. A finite word $w \in A^*$ is in $L_{\Diam^{\a}_Q}$ if, and only if, the content of $w$ on $Q$ does contain the $k$-tuple $\a$, which allows us to write 
\[L_{\Diam^{\a}_Q} = \bigcup_{\a \in B \subseteq A^k \atop B' \subseteq A^k}(K_{Q,B} \cap K_{Q^c,B'})
             = \bigcup_{\a \in B \subseteq A^k \atop B' \subseteq A^k} K_{(Q,Q^{c}),(B,B')}\]
and allows us to conclude.
             
Now, we fix a finite colouring $\mc Q$ of $\mb N^k$ with $\ell \geq 1$ colours and $\B \in \mc P(A^{k})^{\ell}$, and we prove that $K_{\mc Q,\B}$ can be written as a Boolean combination of languages of the form $L_{\Diam^{\a}_Q}$.
First, note that, since \[K_{\mc Q,\B} = \bigcap_{j=1}^{\ell}K_{Q_j,B_j},\]
we only need to prove the result for any $K_{Q,B}$, with $Q \subseteq \mb N^k$ and $B \subseteq A^k$.
Now, we have
\begin{align*}
K_{Q,B} &= \{w \in A^{*} \colon \langle w,Q \rangle=B \}  \\
        &= \{w \in A^{*} \colon \{ \a \in A^k \colon c_{\a}(w) \cap Q \neq \emptyset \} = B \}  \\
		&= (\bigcap_{\a \in B} \{w \in A^{*} \colon c_{\a}(w) \cap Q \neq \emptyset\}) \cap (\bigcap_{\a \notin B} \{w \in A^{*} \colon c_{\a}(w) \cap Q = \emptyset\})  \\
        &= \bigcap_{\a \in B} L_{\Diam^{\a}_Q} \cap \bigcap_{\a \notin B} (L_{\Diam^{\a}_Q})^{c}.
\end{align*}
which allows us to conclude.
\end{proof}

\subsection{A certain family of ultrafilter equations}\label{section:family}

In this section, we introduce a general family of ultrafilter equations on $\beta(A^*)$ which will encompass every ultrafilter equation we will require in order to describe $\mc B_1$. 
We then explain how that it is possible to reformulate these equations into a condition that requires the existence of a certain finite colourings of $\mb N$. This property will be our main tool in order to check soundness and completeness in Section \ref{section:equasig1}.

Let us start by defining the family of ultrafilter equations on $\beta(A^*)$ which will be at the center of our study.
For any $k,n \geq 1$ we use the notation $A^{*} \otimes (\mb N^k)^{n}$ in order to refer to 
$A^{*}\otimes \underbrace{(\mb N^k \times ... \times \mb N^k)}_{n\text{ times}}$ introduced in Section \ref{sect:semantics}.

\begin{definition}\label{def:Euv}
For any $k,n \geq 1$, any finite family of maps 
$p_1,...p_n: A^* \otimes (\mb N^k)^{n} \to \mb N^k $ and 
$u,v : A^* \otimes (\mb N^k)^{n} \to A^*$,
we denote by $\mc E^{p_1,...,p_n}_{u=v}$ the family of ultrafilter equations 
\[\beta u(\nu) \leftrightarrow \beta v(\nu),\] where $\nu$ ranges over all elements of 
$\beta(A^* \otimes (\mb N^k)^{n})$ such that
\[\beta p_1(\nu) = ... = \beta p_{n}(\nu).\]
\end{definition}

The ultrafilter equations we use in order to describe $\mc B_1$
are all particular instances of the ones introduced in Definition \ref{def:Euv}. 
Let us provide some intuition behind these equations. 
Let us consider the Boolean algebra of languages $\mc B \Sigma_1[\mc N_0,\mc N_1^{u}]$ corresponding to Boolean combinations of sentences written by using nullary predicates, unary uniform numerical predicates, and letter predicates. It has been proven in \cite{GKP}, Theorem 5.16, that $\mc B \Sigma_1[\mc N_0,\mc N_1^{u}] \cap Reg$, the Boolean algebra of regular languages in $\mc B \Sigma_1[\mc N_0,\mc N_1^{u}]$ is described by the profinite equations 
\[(x^{\omega -1}s)(x^{\omega -1}t) = (x^{\omega -1}t)(x^{\omega -1}s) \text{  and  } (x^{\omega-1}s)^{2} = x^{\omega-1}s,\] for $x,s,t$ words of the same length.
The profinite monoid on $A^{*}$ is a compactification of $A^*$ which embeds in $\beta(A^*)$, and thus the ultrafilter equations we want to obtain are, in a sense, a generalization of these profinite equations to the setting of ultrafilters.

More precisely, we define the maps that model the operations we are interested in, such as swapping two positions in a finite word, or adding a letter at the end of a finite word, at the set-theoretic level, and we then consider their continuous extension obtained by \v Cech-Stone compactification.

For any $k \geq 1$, any finite word $w \in A^*$, any $k$-tuple of letters $\a \in A^k$ 
and any family of distinct integers $\j \in |w|^k$, we define \[w(\j \rightarrow \a)\] to be the word obtained by replacing, for any $m \in \{1,...,k\}$, $w_{j_{m}}$ by $a_m$ in $w$.
This allows us to define the map $f_{\a}: A^{*} \otimes \mb N^k \to A^*$ as follows: for any $(w,\j)$ in $A^{*} \otimes \mb N^k$,
\[f_{\a}(w,\j):= \begin{cases} w(\j \rightarrow \a) \text{  if all of the $j_m$, for $m \in \{1,...,k\}$, are distinct} \\ w \text{  otherwise  } \end{cases}, \]
and its continuous extension $\beta f_{\a}: \beta(A^{*} \otimes \mb N^k) \to \beta(A^*)$.
Following the ideas introduced in Section \ref{section:colouring}, we see elements of $\beta(\mb N^k)$ as generalized $k$-tuples of position. In particular, it is not equivalent to consider a generalized word with a $k$-tuple of positions, that is, an element of $\beta(A^{*} \otimes \mb N^k)$, and a generalized word with a generalized $k$-tuple of positions, that is, an element of $\beta(A^*) \times \beta(\mb N^k)$. If we want to consider the generalized $k$-tuple of positions associated to $\nu \in \beta(A^* \otimes \mb N^k)$, we look at the ultrafilter $\beta \pi (\nu) \in \beta(\mb N^k)$, where $\pi : A^{*} \otimes \mb N^k \to \mb N^k$ is the canonical projection.

We now introduce the generalization of the profinite equation $(x^{\omega -1}s)(x^{\omega -1}t) = (x^{\omega -1}t)(x^{\omega -1}s)$, for $x,s,t$ words of the same length, to the ultrafilter setting.

\begin{definition}\label{def:Eab=ba}
For any $a,b \in A$, we consider the map $f_{a,b}: A^{*} \otimes \mb N^2 \to A^{*}$, 
which sends any $(w,j_1,j_2) \in A^{*} \otimes \mb N^2$ to
\[f_{a,b}(w,j_1,j_2) := \begin{cases} w((j_1,j_2)\rightarrow (a,b))   \text{  if $j_1 \neq j_2$}  \\ w   \text{  otherwise}  \end{cases}.\]
We denote by $\pi_1,\pi_2: A^{*} \otimes \mb N^2 \to \mb N$ the canonical projections 
on the first and second coordinate. Finally, for any $a,b \in A$, we denote by $\mc E_{ab=ba}$ 
the family of ultrafilter equations $\mc E^{\pi_1,\pi_2}_{f_{a,b},f_{b,a}}$ as in Definition \ref{def:Euv}.
\end{definition}

We introduce the two other family of ultrafilter equations that will be used in section \ref{ection:equasig1}.

\begin{definition}\label{def:Eaab=abb}
For any $a,b \in A$, we consider the map $f_{a,a,b}: A^{*} \otimes \mb N^3 \to A^{*}$, 
which sends any $(w,j_1,j_2,j_3) \in A^{*} \otimes \mb N^3$ to
\[f_{a,a,b}(w,j_1,j_2,j_3) := \begin{cases} w((j_1,j_2,j_3)\rightarrow (a,a,b))   \text{   if $j_1 \neq j_2 \neq j_3$}  \\ w   \text{  otherwise}  \end{cases}.\]
We denote by $\pi_1,\pi_2,\pi_3: A^{*} \otimes \mb N^3 \to \mb N$ the canonical projections 
on the first, second and third coordinate. Finally, for any $a,b \in A$, we denote by $\mc E_{aab=abb}$ 
the family of ultrafilter equations $\mc E^{\pi_1,\pi_2,\pi_3}_{f_{a,a,b},f_{a,b,b}}$ 
as in Definition \ref{def:Euv}.
\end{definition}

\begin{definition}\label{def:Ea=a.a}
For any $a \in A$, we consider the map $f_a: A^{*} \otimes \mb N \to A^{*}$ which sends any finite word with a marked position $(w,i)$ to $w(i \rightarrow a)$, and the map $f_{a}.a: A^{*} \otimes \mb N \to A^{*}$ which sends any finite word with a marked position $(w,i)$ to $w(i \rightarrow a).a$.
\end{definition}
We denote by $\pi: A^{*} \otimes \mb N \to \mb N$ the canonical projection and $|\cdot|:A^{*} \otimes \mb N \to \mb N$ the map that sends $(w,i) \in A^{*} \otimes \mb N$ to $|w|$.
Finally, for any $a \in A$, we denote by $\mc E_{a=a.a}$ 
the family of ultrafilter equations $\mc E^{\pi,|\cdot|}_{f_{a},f_{a}.a}$ as in Definition \ref{def:Euv}.


Our goal in the upcoming sections will be to prove the following results.

\begin{theorem}\label{theorem:eunaire}
A language $L \subseteq A^*$ is in $\mc B_1$
if, and only if,
$L$ satisfies the families of ultrafilter equations $\mc E_{ab=ba}$, $\mc E_{aab=abb}$ and $\mc E_{a=a.a}$, 
for every $a,b \in A$.
\end{theorem}

\subsection{Ultrafilter equations in terms of finite colourings}

In this subsection, we present a reformulation of the equations introduced in Definition \ref{def:Euv} in terms of finite colourings. First, we prove a technical lemma that allows for a rephrasing of one of the conditions involved in Definition \ref{def:Euv}.

\begin{lemma}\label{lemma:betapi}
For any $k,n \geq 1$, consider an ultrafilter $\nu \in \beta(A^{*} \otimes (\mb N^k)^{n})$ and a family of $n$ maps $p_1,...,p_n: A^{*} \otimes (\mb N^k)^{n} \to \mb N^k$. For any $\alpha \in \beta(\mb N^k)$, the following statements are equivalents.
\begin{enumerate}
\item For every $j \in \{1,...,n\}$, $\beta p_{j}(\nu) = \alpha$.
\item For every $Q \in \alpha$, $\bigcap_{j=1}^{n} p_{j}^{-1}(Q) \in \nu$. 
\end{enumerate}
Furthermore, these conditions hold for $\nu$ with respect to some $\alpha$ if, and only if, for every finite colouring $\mc Q=(Q_1,...,Q_{\ell})$ of $\mb N^k$, where $\ell \geq 1$, we have \[ \bigcup_{i=1}^{\ell}\bigcap_{j=1}^{n} p_{j}^{-1}(Q_i) \in \nu.\]
In particular, the family of subsets of the form $\bigcup_{i=1}^{\ell}\bigcap_{j=1}^{n} p_{j}^{-1}(Q_i)$, for every finite colouring $\mc Q=(Q_1,...,Q_{\ell})$ of $\mb N^k$, where $\ell \geq 1$, forms a filter subbasis of $\nu$.
\end{lemma}

\begin{proof}
For (1) ìmplies (2), let us assume that, for every $j \in \{1,...,n\}$, $\beta p_{j}(\nu) = \alpha$. Then, for every $j \in \{1,...,n\}$, and every $Q \in \alpha$, we have $p_{j}^{-1}(Q) \in \nu$, which implies that $\bigcap_{j=1}^{n} p_{j}^{-1}(Q) \in \nu$. \newline For (2) implies (1), fixing $j \in \{1,...,n\}$, we prove that $\alpha \subseteq \beta p_{j}(\nu)$, which is enough to prove that they are equal since ultrafilters are maximal for inclusion.
Fix $Q \in \alpha$. By (2) we have that $\bigcap_{k=1}^{n} p_{k}^{-1}(Q) \in \nu$, and since $\bigcap_{k=1}^{n} p_{k}^{-1}(Q) \subseteq p_{j}^{-1}(Q)$, we deduce by up-set that $p_{j}^{-1}(Q) \in \nu$, and thus that $Q \in \beta p_{j}(\nu)$. \newline We now treat the last assertion. On the one hand assume that there exists $\alpha \in \beta(\mb N^k)$ satisfying (2). Since $\alpha$ is an ultrafilter, for every finite colouring $\mc Q=(Q_1,...,Q_{\ell})$ of $\mb N^k$, where $\ell \geq 1$, there exists $k \in \{1,...,\ell\}$ such that $Q_k \in \alpha$. Therefore, $\bigcap_{j=1}^{n}p_{j}^{-1}(Q_k) \in \nu$, and since
\[\bigcap_{j=1}^{n}p_{j}^{-1}(Q_k) \subseteq  \bigcup_{i=1}^{\ell}\bigcap_{j=1}^{n} p_{j}^{-1}(Q_i),\]
we deduce by up-set that this last subset belongs to $\nu$. \newline On the other hand, assume that for every finite colouring $\mc Q=(Q_1,...,Q_{\ell})$ of $\mb N^k$, where $\ell \geq 1$, $\bigcup_{i=1}^{\ell}\bigcap_{j=1}^{n} p_{j}^{-1}(Q_i)$ is in $\nu$. We need to prove the existence of an ultrafilter $\alpha \in \beta(\mb N^k)$ such that (2) holds. We set 
\[ \alpha := \{ Q \subseteq \mb N^k \colon  \bigcap_{j=1}^n p_{j}^{-1}(Q) \in \nu \}.\] 
We prove that $\alpha$ is an ultrafilter.
First, since $\nu$ is an ultrafilter, it does not contain the empty set, and thus $\alpha$ does not contain the empty-set. Also, since inverse image preserves finite intersections and inclusion, we deduce that $\alpha$ is indeed a filter on $\mb N^k$.
Furthermore, for any $Q \subseteq \mb N^k$, $(Q,Q^{c})$ is a finite colouring of $\mb N^k$, thus, by (3),
\[ \bigcap_{j=1}^n p_{j}^{-1}(Q) \cup \bigcap_{j=1}^n p_{j}^{-1}(Q^{c}) \in \nu.\]
This union being disjoint, and $\nu$ being an ultrafilter, we deduce that exactly one element in $\{\bigcap_{j=1}^n p_{j}^{-1}(Q),\bigcap_{j=1}^n p_{j}^{-1}(Q)\}$ belongs to $\nu$, in other words exactly one element in $\{Q,Q^c\}$ is in $\alpha$. \newline Finally, in order to prove that we have a filter subbasis, we prove that 
\[ \{\bigcup_{i=1}^{\ell}\bigcap_{j=1}^{n} p_{j}^{-1}(Q_i) \colon (Q_1,...,Q_{\ell}) \text{ is a finite colouring of $\mb N^k$}, \ell \geq 1  \} \] is closed under finite intersection.
Fix $(Q_1,...,Q_{\ell})$ and $(Q'_1,...,Q'_{\ell'})$ two finite colourings of $\mb N$, for some $\ell, \ell' \geq 1$. If we consider $(Q''_1,...,Q''_{\ell''})$, with $\ell'' \geq 1$, a finite colouring which refines both of these colourings, we obtain that 
\[ \bigcup_{i=1}^{\ell}\bigcap_{j=1}^{n} p_{j}^{-1}(Q_i) \cap 
 \bigcup_{i=1}^{\ell'}\bigcap_{j=1}^{n} p_{j}^{-1}(Q'_i) 
 =  \bigcup_{i=1}^{\ell''}\bigcap_{j=1}^{n} p_{j}^{-1}(Q''_i),\]
which proves the that the family we mentioned forms a filter subbasis.
\end{proof}

We now treat the other condition involved in the definition of the family of ultrafilter equations $\mc E^{p_1,...,p_n}_{u=v}$ from Definition \ref{def:Euv},
that is \[L \models (\beta u(\nu) \leftrightarrow \beta v(\nu)).\]
For any set $S$, and any two subsets $T_1,T_2$ of $S$ we denote by $T_1 \Delta T_2$ their \emph{symmetric difference}, that is the subset of $S$ such that, for any $s \in S$, $s \notin T_1 \Delta T_2$ if, and only if the condition \[s \in T_1 \Longleftrightarrow s \in T_2 \] holds.

\begin{lemma}\label{lemma:secondpart}
For any set $S$, any two subsets $T_1, T_2 \subseteq S$ and any ultrafilter $\gamma \in \beta(S)$, the following conditions are equivalent.
\begin{enumerate}
\item $T_1 \in \gamma$ if, and only if, $T_2 \in \gamma$.
\item $(T_1 \Delta T_2)^{c} \in \gamma$.
\end{enumerate} 
In particular, for any $k,n \geq 1$, consider an ultrafilter $\nu \in \beta(A^{*} \otimes (\mb N^k)^{n})$, two maps 
$u,v: A^{*} \otimes (\mb N^k)^{n} \to A^*$,
and a language $L \subseteq A^{*}$. Then, $L$ satisfies the ultrafilter equation 
$\beta u(\nu) \leftrightarrow \beta v(\nu)$, if, and only if $E_{L,u,v} \in \nu$, where
\[E_{L,u,v} := \{(w,\i_{1},...,\i_{n}) \in A^{*} \otimes (\mb N^k)^n \colon 
u(w,\i_{1},...,\i_{n}) \in L \Longleftrightarrow v(w,\i_{1},...,\i_{n}) \in L \}. \]
\end{lemma}

\begin{proof}
Let us assume that $T_1 \in \gamma$ if, and only if, $T_2 \in \gamma$.
Since $\gamma$ is closed under finite intersections, this last statement is
equivalent to saying that 
$T_1 \cap T_2 \in \gamma$ or $T_1^{c} \cap T_2^{c} \in \gamma,$
and thus, since $\gamma$ is an ultrafilter, equivalent to 
$(T_1 \cap T_2) \cup (T_1^{c} \cap T_2^{c}) \in \gamma$.
By definition of the symmetric difference,
\[T_1 \Delta T_2 = ((T_1 \cap T_2) \cup (T_1^{c} \cap T_2^{c}))^{c},\]
thus this is equivalent to saying that $(T_1 \Delta T_2)^{c} \in \gamma$. \newline
The final statement is a simple application of this result for $S= A^{*} \otimes (\mb N^k)^n$, $T_1 = u^{-1}(L)$ and $T_2 = v^{-1}(L)$.
\end{proof}

We are now ready to give a reformulation of the equations of Definition \ref{def:Euv} in terms of a condition relative to finite colourings of $\mb N^k$.

\begin{proposition}\label{prop:reformulationEuv}
For any $k,n \geq 1$, any maps
$p_1,...,p_n : A^{*} \otimes (\mb N^k)^{n}  \to \mb N^k$
and  
$u,v: A^{*} \otimes (\mb N^k)^{n} \to A^*$, a language $L \subseteq A^*$ satisfies 
$\mc E^{p_1,...,p_n}_{u=v}$ if, and only if, there exists a finite colouring 
$\mc Q = (Q_1,...,Q_{\ell})$ of $\mb N^k$ for some $\ell\geq 1$ such that
\[ \bigcup_{i=1}^{\ell} \bigcap_{j=1}^{n} p_{j}^{-1}(Q_i) \subseteq E_{L,u,v}.\]
\end{proposition}

\begin{proof}
For any language $L \subseteq A^{*}$, to satisfy the ultrafilter equation
$\mc E^{p_1,...,p_n}_{u=v}$ amounts to the following condition.

\[\forall \nu \in \beta(A^{*} \otimes (\mb N^k)^n),
[ \beta p_{1}(\nu)=...=\beta p_{n}(\nu) \Longrightarrow 
( L \models \beta u (\gamma) \leftrightarrow \beta v (\gamma))].\]

By applying Lemma \ref{lemma:betapi} and Lemma \ref{lemma:secondpart}, we can reformulate this condition as follows.

\[\forall \nu \in \beta(A^{*} \otimes (\mb N^k)^n),
(\{ \bigcup_{i=1}^{\ell} \bigcap_{j=1}^{n} p_{j}^{-1}(Q_i) \colon (Q_1,...,Q_{\ell}) \text{ is a colouring of $\mb N^k$}, \ell\geq 1  \} \subseteq \nu  \Longrightarrow 
E_{L,u,v} \in \nu) \]

In particular, we know that a filter is the intersection of all of the ultrafilters which contain it. Therefore, since Lemma \ref{lemma:betapi} allows us to consider the filter 
\[ \mc F := {\uparrow} \{ \bigcup_{i=1}^{\ell} \bigcap_{j=1}^{n} p_{j}^{-1}(Q_i) \colon (Q_1,...,Q_{\ell}) \text{ is a finite colouring of $\mb N^k$, }  \ell \geq 1  \} ,\]
we can simplify our condition into $E_{L,u,v} \in \mc F$, that is, there exists a finite colouring 
$\mc Q = (Q_1,...,Q_{\ell})$ of $\mb N^k$ for some $\ell \geq 1$ such that
\[ \bigcup_{i=1}^{\ell} \bigcap_{j=1}^{n} p_{j}^{-1}(Q_i) \subseteq E_{L,u,v}.\]
\end{proof}

Let us assume that a language $L$ satisfies two families of ultrafilter equations, $\mc E^{p_1,...,p_{n}}_{u=v}$ and $\mc E^{p'_1,...,p'_{n'}}_{u'=v'}$. 
By applying Proposition \ref{prop:reformulationEuv}, this is equivalent to assuming the existence of two finite colourings $\mc Q=(Q_1,...,Q_{\ell})$ and $\mc Q'=(Q'_1,...,Q'_{\ell'})$, for some $\ell, \ell' \geq 1$, such that 

\[\bigcup_{i=1}^{\ell} \bigcap_{j=1}^{n} p_{j}^{-1}(Q_i) \subseteq E_{L,u,v} \]
and
\[\bigcup_{i=1}^{\ell'} \bigcap_{j=1}^{n'} (p'_{j})^{-1}(Q'_i) \subseteq E_{L,u',v'}.\]
 
Now, considering a common refinement $\mc Q''=(Q''_1,...,Q''_{\ell''})$ of the two colourings $\mc Q$ and $\mc Q'$, for some $\ell'' \geq 1$, we obtain in particular that
\[\bigcup_{i=1}^{\ell''} \bigcap_{j=1}^{n} p_{j}^{-1}(Q''_i) \subseteq E_{L,u,v} \]
and
\[\bigcup_{i=1}^{\ell''} \bigcap_{j=1}^{n} (p'_{j})^{-1}(Q''_i) \subseteq E_{L,u',v'}.\]

This allows for a reformulation of the ultrafilter equations introduced in Definition \ref{def:Eab=ba}, \ref{def:Eaab=abb} and \ref{def:Ea=a.a}.

\begin{corollary}\label{cor:reformulationequation}
A language $L \subseteq A^*$ satisfies the families of ultrafilter equations $\mc E_{ab=ba}, \mc E_{aab=abb}$ and $\mc E_{a=a.a}$, for every $a,b \in A$, if, and only if, there exists a finite colouring $\mc Q = (Q_1,...,Q_{\ell})$ of $\mb N$, for some $\ell\geq 1$, such that, for every $a,b \in A$,
\[\begin{cases}
\bigcup_{i=1}^\ell A^{*} \otimes Q_{i}^{2} \subseteq E_{L,f_{a,b},f_{b,a}} \\ \bigcup_{i=1}^\ell A^{*} \otimes Q_{i}^{3} \subseteq E_{L,f_{a,a,b},f_{a,b,b}} \\ \bigcup_{i=1}^\ell L_{Q_{i}} \otimes Q_{i} \subseteq E_{L,f_{a},f_{a.a}}
\end{cases},\]
where, for every $i \in \{1,...,\ell\}$, \[L_{Q_{i}} \otimes Q_{i} := \{(w,j) \in A^{*} \otimes \mb N \colon j \in Q_i \text{ and } |w| \in Q_i \}.\]
\end{corollary}

\begin{example}\label{example:sound1}
As a first application of these reformulations in terms of finite colouring of $\mb N$, we check that the couples of ultrafilters in the family $\mc E_{ab=ab}$, for every $a,b \in A$ are indeed in the kernel of the continuous quotient $q:\beta(A^*) \twoheadrightarrow X_1$, dual to the canonical embedding $\mc B_1 \hookrightarrow \mc P(A^*)$.

Recall that we proved in Proposition \ref{prop:carcImultrafilter} that the map $q$ sends any ultrafilter $\gamma \in \beta(A^*)$ to the family of closed subsets $(C_d(\gamma))_{d \in A}$, where, for every $d \in A$,
\[C_{d}(\gamma):= \bigcap_{L_{\Box^{d}_Q} \in \gamma}\wh{Q}.\]
Now, fix $a,b \in A$, and consider $\nu \in \beta(A^* \otimes \mb N^2)$, such that $\beta \pi_1(\nu)=\beta \pi_2(\nu)$: we prove that $\beta f_{a,b}(\nu)$ and $\beta f_{b,a}(\nu)$ have the same image under $q$.
Fix $d \in A$. First, note that, for every $d \in A$,
\begin{align*}
L_{\Box^{d}_Q} \in \beta f_{a,b} (\nu)
&\Longleftrightarrow  f_{a,b}^{-1}(L_{\Box^{d}_Q}) \in \nu \\ 
&\Longleftrightarrow \{(w,j_1,j_2) \in A^{*} \otimes \mb N^2 \colon w((j_1,j_2) \rightarrow (a,b)) \in L_{\Box^{d}_Q} \} \in \nu \\
&\Longleftrightarrow \{(w,j_1,j_2) \in A^{*} \otimes \mb N^2 \colon c_{d}(w((j_1,j_2) \rightarrow (a,b))) \subseteq Q \} \in \nu.
\end{align*}

Now, since $\beta \pi_1(\nu)=\beta \pi_2(\nu)$, we use Lemma \ref{lemma:betapi} with the colouring $(Q,Q^c)$ of $\mb N$, and we obtain that \[(A^{*}\otimes Q^2) \cup (A^{*}\otimes (Q^{c})^2) \in \gamma.\] 
Therefore, by intersection, \[\{(w,j_1,j_2) \in A^{*} \otimes \mb N^2 \colon c_{d}(w((j_1,j_2) \rightarrow (a,b))) \subseteq Q \} \in \nu \] is equivalent to 
\begin{align*}
\{(w,j_1,j_2) \in A^{*} \otimes \mb N^2 \colon j_1,j_2 \in Q \text{  and  } c_{d}(w((j_1,j_2) \rightarrow (a,b))) \subseteq Q \}   \cup \\
\{(w,j_1,j_2) \in A^{*} \otimes \mb N^2 \colon j_1,j_2 \in Q^{c} \text{  and  } c_{d}(w((j_1,j_2) \rightarrow (a,b))) \subseteq Q^{c} \}\in \nu.
\end{align*}
Finally, observe that, for any $j_1,j_2 \in Q$, \[c_{d}(w((j_1,j_2) \rightarrow (a,b))) = c_{d}(w((j_1,j_2) \rightarrow (b,a))),\]
and the same holds for any $j_1,j_2 \in Q^{c}$. This allows us to prove that, for any $d \in A$,
\[L_{\Box^{d}_Q} \in \beta f_{a,b} (\nu) \Longleftrightarrow L_{\Box^{d}_Q} \in \beta f_{b,a}(\nu),\]
and thus, for every $d \in A$, $C_{d}(\beta f_{a,b}(\nu))= C_{d}(\beta f_{a,b}(\nu))$ which means that $q(\beta f_{a,b}(\nu)) = q(\beta f_{b,a}(\nu))$.
\end{example}

\section{A topological proof for of soundness and completness}\label{section:equasig1}

In this section, we explain how it is possible to use Corollary \ref{cor:reformulationequation} in order to show that the ultrafilter equations we introduced in Definition \ref{def:Eaab=abb}, \ref{def:Eaab=abb} and \ref{def:Ea=a.a} allow for a description of $\mc B_{1}$.

\subsection{Soundness}

We start by proving that the Boolean algebra $\mc B_{1}$ satisfy the family of equations we introduced in Section \ref{section:family}. The reformulation of these equations we gave in Corollary \ref{cor:reformulationequation}, combined with the knowledge of the family of generators that we introduced in Proposition \ref{prop:caraB_kKQB}, enable us to check it in a straight-forward fashion.

\begin{proposition}\label{prop:soundness}
Let $L \subseteq A^*$ be a language: if $L$ is in $\mc B_{1}$, then $L$ satisfies the ultrafilter equations in the families $\mc E_{ab=ba}$, $\mc E_{aab=abb}$ and $\mc E_{a=a.a}$, for every $a,b \in A$.
\end{proposition}

\begin{proof}
(1): By Proposition \ref{prop:caraB_kKQB}, we know that $\mc B_{1}$ is the Boolean algebra generated by the languages $K_{\mc Q,\B}$, where $\mc Q$ ranges over all finite colourings of $\mb N$ with $\ell$ colours and $\B$ ranges over $\mc P(A)^{\ell}$, where $\ell$ ranges overs $\mb N_{>0}$. Therefore, it is enough to prove that these languages satisfy the ultrafilter equations in question.
We do so by using the reformulation of these equations introduced in Corollary \ref{cor:reformulationequation}.
Fix a finite colouring $\mc Q$ of $\mb N$ with $\ell$ colours, where $\ell \geq 1$, and fix $\B \in \mc P(A)^{\ell}$. Fix $a,b \in A$. \newline \underline{$\mc E_{ab=ba}$:}
Let us start by proving that 
\[ \bigcup_{i=1}^{\ell}A^{*} \otimes Q_i^2 \subseteq E_{K_{\Q,\B},f_{a,b},f_{b,a}}.\]
Fix $i \in \{1,...,\ell\}$. We prove that every $(w,j_1,j_2) \in A^{*}\otimes Q_i^2$ belongs to $E_{K_{\Q,\B},f_{a,b},f_{b,a}}$.
On the one hand, notice that, for every $i' \in \{1,...,\ell\}$, with $i \neq i'$, the content of $w((j_1,j_2)\rightarrow(a,b))$ on $Q_{i'}$ is the same than the content of 
$w((j_1,j_2)\rightarrow(b,a))$
on $Q_{i'}$: indeed, they are both equal to the content of $w$ on $Q_{i'}$.
On the other hand, the content of $w((j_1,j_2)\rightarrow(a,b))$ on $Q_{i}$ is the same than the content of 
$w((j_1,j_2)\rightarrow(b,a))$
on $Q_{i}$: we only switched the letters at position $j_1$ and $j_2$, which does not add nor remove any letter out of the content on $Q_i$.
We deduce that, for every $i \in \{1,...,\ell\}$, $A^{*} \otimes Q_i^2 \subseteq E_{K_{\Q,\B},f_{a,b},f_{b,a}}$, thus we conclude by Corollary \ref{cor:reformulationequation} that $K_{\Q,\B}$ satisfies the ultrafilter equation $\mc{E}_{ab=ba}$. \newline \underline{$\mc E_{aab=abb}$:}
The argument is almost exactly the same to prove that $K_{\mc Q,\B}$ satisfies the ultrafilter equation $\mc{E}_{aab=abb}$. Fix $i \in \{1,...,\ell\}$. We prove that every $(w,j_1,j_2,j_3) \in A^{*}\otimes Q_i^3$ belongs to $E_{K_{\Q,\B},f_{a,a,b},f_{a,b,b}}$. On the one hand, notice that, for every $i' \in \{1,...,\ell\}$, with $i \neq i'$, the content of $w((j_1,j_2,j_3)\rightarrow(a,a,b))$ on $Q_{i'}$ is the same than the content of 
$w((j_1,j_2,j_3)\rightarrow(a,b,b))$
on $Q_{i'}$: indeed, they are both equal to the content of $w$ on $Q_{i'}$.
On the other hand, the content of $w((j_1,j_2,j_3)\rightarrow(a,a,b))$ on $Q_{i}$ is the same than the content of 
$w((j_1,j_2,j_3)\rightarrow (a,b,b))$ on $Q_{i}$: we only replaced the occurrence of $a$ available at the position $j_2 \in Q_i$ by an occurrence of $b$, which was already available at the position $j_3 \in Q_i$. The letter $a$ is still available at the position $j_1 \in Q_i$: the content on $Q_i$ has not been altered.
We deduce that
\[\bigcup_{i=1}^{\ell} A^{*}\otimes Q_i^3 \subseteq E_{K_{\Q,\B},f_{a,a,b},f_{a,b,b}},\]
and we conclude  by Corollary \ref{cor:reformulationequation} that $K_{\Q,\B}$ satisfies the ultrafilter equation $\mc{E}_{aab=abb}$. \newline \underline{$\mc E_{a=a.a}$:} Finally, we prove that
\[\bigcup_{i=1}^{\ell} L_{Q_i}\otimes Q_i \subseteq E_{K_{\Q,\B},f_{a},f_{a}.a}.\]
Fix $i \in \{1,...,\ell\}$. We prove that every $(w,j) \in A^{*}\otimes Q_i$ such that $|w| \in Q_i$ belongs to $E_{K_{\Q,\B},f_{a},f_{a}.a}$.
On the one hand, notice that, for every $i' \in \{1,...,\ell\}$, with $i \neq i'$, the content of $w(j\rightarrow a)$ on $Q_{i'}$ is the same than the content of $w(j \rightarrow a).a$ on $Q_{i'}$, since $|w|$ does not belong to $Q_{i'}$.
On the other hand, the content of $w(j\rightarrow a)$ on $Q_{i}$ is the same than the content of $w(j\rightarrow a).a$ on $Q_i$: the only difference between these two words is the presence of one more occurrence of $a$ on a position which belongs to $Q_i$, but since $j \in Q_i$, the letter $a$ is already present in the content of both oh these words on $Q_i$. We conclude that $K_{\Q,\B}$ satisfies the ultrafilter equation $\mc{E}_{a=a.a}$.
\end{proof}

\subsection{Completeness for $\mc B_1$}

We prove that the languages satisfying the family of ultrafilter equations $\mc E_{ab=ba}$, $\mc E_{aab=abb}$ and $\mc E_{a=a.a}$, for every $a,b \in A$, all belong to $\mc B_1.$ In order to do so, we use the presentation of $\mc B_1$ we provided in Proposition \ref{prop:caraB_kKQB}. As explained in Section \ref{section:complementBk}, the languages generating $\mc B_1$ are of the form $K_{\mc Q,\B}$, where $\mc Q$ is a finite colouring of $\mb N$ with $\ell$ colours, for some $\ell \geq 1$, and $\B \in \mc P(A)^\ell$. For a fixed finite colouring $\mc Q$, the finitely indexed equivalence relation $\sim_{\mc Q}$ introduced in Section \ref{section:complementBk} gives a finite partition of $A^{*}$. We prove that, if a language $L \subseteq A^*$ satisfies all of our ultrafilter equations, then it is possible to find a finite colouring $\mc Q$ such that $L$ can be written as a finite union of some of the equivalence classes for $\sim_{\mc Q}$.

We start by detailing the situation when the colouring in question only has two colours, and we then generalize the argument to any finite colouring of $\mb N$.

\begin{proposition}\label{prop:transductionN1ucasfacile}
Let $L \subseteq A^*$ be a language satisfying $\mc E_{ab=ba}$, $\mc E_{aab=abb}$ and $\mc E_{a=a.a}$, for every $a,b \in A$. Let us consider a finite colouring $\mc Q$ of $\mb N$ satisfying the condition from Corollary \ref{cor:reformulationequation}, and assume that $\mc Q=(Q,Q^c)$ for some $Q \subseteq \mb N$.
Then, for any words $w, w' \in A^*$ such that $w \sim_{Q} w'$ and, for every $i \in Q^c \cap \{0,...,min(|w|,|w'|)-1\}$, $w_i=w'_{i}$, we have that 
\[w \in L \Longleftrightarrow w' \in L.\]
\end{proposition}
\begin{proof}\label{proof:transductionN1ucasfacile}
We consider a language $L \subseteq A^*$ such that there exists $Q \subseteq \mb N$ such that, for every $a,b \in A$,
\[ (A^{*}\otimes Q^2) \cup( A^{*} \otimes (Q^c)^2) \subseteq E_{L,f_{a,b},f_{b,a}},\]
\[ (A^{*}\otimes Q^3) \cup (A^{*} \otimes (Q^c)^3) \subseteq E_{L,f_{a,a,b},f_{a,b,b}},\]
and
\[ (L_{Q} \otimes Q) \cup (L_{Q^{c}} \otimes Q^{c}) \subseteq E_{L,f_{a},f_{a}.a}.\]
We consider two finite words $w,w' \in A^*$ such that $w \sim_{Q} w'$ and for every $i \in Q^c \cap \{0,...,min(|w|,|w'|)-1\}$, $w_i=w'_{i}$.
We set $N:=min(|w|,|w'|)$.
We want to prove that $w$ is in $L$ if, and only if, $w'$ is in $L$. 
In order to do so, we start by defining the following families of endofunctions of $A^{*}$.

\[\mc M_{\mc E,1,Q}^{ab=ba}:= 
\{(f_{a,b}^{j_1,j_2},f_{b,a}^{j_1,j_2}) \colon j_1,j_2 \in Q \cap \{0,...,N-1\}, a,b \in A \},\]
where, for every $j_1,j_2 \in Q \cap \{0,...,N-1\}$, and every $a,b \in A$, $f_{a,b}^{j_1,j_2}: A^{*} \to A^{*}$ sends any finite word $v$ to $f_{a,b}(v,j_1,j_2)$ if $(v,j_1,j_2)$ is in $A^{*}\otimes \mb N^2$, and to $v$ otherwise. We define in an analogous way

\[\mc M_{\mc E,1,Q}^{aab=abb}:= 
\{(f_{a,a,b}^{j_1,j_2,j_3},f_{a,b,b}^{j_1,j_2,j_3}) \colon j_1,j_2,j_3 \in Q \cap \{0,...,N-1\}, a,b \in A \},\]

\[\mc M_{\mc E,1,Q}^{a=a.a}:= 
\{(f_{a}^{j_1},f^{j_1}_{a}.a) \colon j_1 \in Q \cap \{0,...,N-1\}, a \in A \},\]

and 
\[\mc M_{\mc E,1,Q}:= \mc M_{\mc E,1,Q}^{ab=ba} \cup \mc M_{\mc E,1,Q}^{aab=abb} \cup \mc M_{\mc E,1,Q}^{a=a.a}.\] 

Now, let us assume that there exists $n \in \mb N$, a family of pairs of maps $(f_m,g_m)_{1 \leq m \leq n}$ and a finite sequence of words $(v_m)_{1 \leq m \leq n}$ satisfying the following properties:

\begin{center}
$\left.\begin{minipage}{11cm}
\begin{enumerate}
\item $v_1 = w$ and $v_{n}=w'$;
\item for every $m \in \{1,...,n\}, (f_m,g_m) \in \mc M_{\mc E,1,Q}$
\item for every $m \in \{1,...,n\}$, $f_{m}(v_m)=v_m$;
\item for every $m \in \{1,...,n-1\}$, $g_{m}(v_{m})=v_{m+1}$, and $g_n(v_n)=v_n$.
\end{enumerate}
\end{minipage}\right\rbrace$ $(\ast)$
\end{center}

In particular, by composing the maps $(g_n,g_{n-1},...,g_1)$, we would have that \[w' = (\mathop{\bigcirc}\limits_{m = n}^{1} g_m)(w),\] and therefore by Corollary \ref{cor:reformulationequation}, we are able to conclude that $w$ is in $L$ if, and only if $w'$ is in $L$.
The end of this proof describes how it is possible to construct all of these elements satisfying $(\ast)$. 

We start by treating the case where $w$ and $w'$ have the same length, that is $N$.
In that case, there exists $n_1,...,n_p < N$ for some $p \in \mb N$ such that
\[Q \cap \{0,...,N-1\} = \{n_1,...,n_p\},\]
and we consider the finite words 
$\restriction{w}{Q} := w_{n_{1}}...w_{n_{p}}$ and 
$\restriction{w'}{Q} := w'_{n_{1}}...w'_{n_{p}}$.
We assumed that $w \sim_Q w'$, thus we set $B:=\langle w,Q \rangle= \langle w',Q \rangle$. Notice that $\restriction{w}{Q}$ and $\restriction{w'}{Q}$ are equal if, and only if, $w$ and $w'$ are equal.

First, assume that each letter in $B$ occurs exactly once in 
$\restriction{w}{Q}$ and $\restriction{w'}{Q}$.
In that case, these two words only differ by the order of their letters, and so do $w$ and $w'$.
If $\restriction{w}{Q}$ and $\restriction{w'}{Q}$ are the same word, then we are done.
Otherwise, there are at least two distinct letters $a,b \in B$ which occur in $\restriction{w}{Q}$ and $\restriction{w'}{Q}$, and yet do not occupy the same position in both words. We thus consider the unique integers $s_1,s_2$ and $t_1,t_2$ in $\{n_1,...,n_p\}$ such that 
$w_{n_{s_1}}=w_{n_{t_1}}=a$  and
$w'_{n_{s_2}}=w'_{n_{t_2}}=b$, and we know that $(s_1,s_2) \neq (t_1,t_2)$.
Then, 
$w=f_{a,b}(w,n_{s_1},n_{t_1})$ and the word $w_1=f_{b,a}(w,n_{s_1},n_{t_1})$ is such that the positions of the occurrences of $a$ in $w_1$ coincide with the one of $w'$.
Since any permutation on a finite set can be written as a finite product of transpositions, we can iterate this reasoning finitely many times, and eventually construct a finite family 
$(f_m,g_m)_{1 \leq m \leq n}\subseteq \mc M_{\mc E,1,Q}^{ab=ba}$ and a finite sequence of words $(v_m)_{1 \leq m \leq n}$
for some $n \in \mb N$ satisfying $(\ast)$. \newline Now, assume there exists a letter which occurs strictly more than once in $\restriction{w}{Q}$ or $\restriction{w'}{Q}$, say $\restriction{w}{Q}$ without any loss of generality. If, for every $a \in B$, $|\restriction{w}{Q}|_{a}=|\restriction{w'}{Q}|_{a}$, then these two words only differ by the order of apparition of each of their letters. By applying a similar reasoning to the case where $\restriction{w}{Q}$ and $\restriction{w'}{Q}$ have exactly one occurrence of each letter in $A$, we can construct a finite family 
$(f_m,g_m)_{1 \leq m \leq n}\subseteq \mc M_{\mc E,1,Q}^{ab=ba}$ and a finite sequence of words $(v_m)_{1 \leq m \leq n}$, for some $n \in \mb N$ satisfying the four conditions we mentioned.
Else, there exists a letter $a \in B$ such that $|\restriction{w}{Q}|_{a} \neq |\restriction{w'}{Q}|_{a}$,
say $|\restriction{w}{Q}|_{a} > |\restriction{w'}{Q}|_{a}$ without any loss of generality.
Pick $n_{s_1}, n_{s_2} \in \{n_1,...,n_p\}$ distincts such that $w_{j_{1}}= w_{j_{2}}=a$.
Now, since $\restriction{w}{Q}$ and $\restriction{w'}{Q}$ have the same length, there exists a letter $b \in B$ with $|\restriction{w}{Q}|_{b} < |\restriction{w'}{Q}|_{b}$, and $b \neq a$. 
Of course, since $w \sim_{Q} w'$, there exists $n_{s_3} \in \{n_1,...,n_p\}$ such that $w_{j_{3}}=b$. 
Then, $w=f_{a,a,b}(w,n_{s_1},n_{s_2},n_{s_3})$ and the word $w_1:=f_{a,b,b}(w,n_{s_1},n_{s_2},n_{s_3})$ satisfies the following properties:
\begin{itemize}
\item $w \sim_{Q} w_1$;
\item for every $i \in Q^{c} \cap \{0,...,N-1\}$, $(w_{1})_{i} = w_i$; 
\item by Corollary \ref{cor:reformulationequation}, $w$ is in $L$ if, and only if, $w_1$ is in $L$;
\item $|w_1|_{a} = |w|_{a}-1 $.
\end{itemize}
We can keep on iterating this reasoning finitely many times, until we obtain a word which brings us back to the case where $|\restriction{w}{Q}|_{a}=|\restriction{w'}{Q}|_{a}$.
By applying this argument to every letter of the alphabet $A$, we construct a finite word $w'' \in A^*$ such that $w'' \sim_{\Q}w'$, for every $i \leq N$, $w''_i=w'_i$ and, for every $a \in B$, $|\restriction{w''}{Q}|_{a}=|\restriction{w'}{Q}|_{a}$: this situation has already been treated previously.

Finally, let us treat the case where $w$ and $w'$ do not have the same length.
Without any loss of generality, let us assume that $|w|<|w'|$.
Recall that we assumed that $w \sim_Q w'$. 
Since $(Q,Q^c)$ is a colouring of $\mb N$, we make a case distinction depending on which color $|w|$ belongs to.
Pick $m \in Q \cap \{0,...,|w|-1\}$, and set
\[a:=  \begin{cases}w_m  \text{   if $|w| \in Q$} \\ w'_{|w|} \text{   if $|w| \in Q^c$} \end{cases}. \]
and
\[j:=  \begin{cases}m  \text{   if $|w| \in Q$} \\ |w| \text{   if $|w| \in Q^c$} \end{cases}. \]
Since $(w,j)$ belongs to $(L_{Q} \otimes Q) \cup (L_{Q^{c}} \otimes Q^{c})$, we have by Corollary \ref{cor:reformulationequation} that $w$ is in $L$ if, and only if, $w.a$ is in $L$. Observe that, by construction, $w \sim_{\Q} w.a$, and for every $i \in Q^c$, $(w.a)_i=w'_i$. By iterating this argument finitely many times, we are able to construct a finite word $w'' \in A^*$ which has the same length than $w'$, and such that $w'' \sim_{\Q} w'$.
We can now apply the reasoning we conducted in the first part of the proof to $w''$ and $w'$ in order to conclude.
\end{proof}

We now generalize the argument to any finite colouring of $\mb N$.

\begin{proposition}\label{prop:transductionN1u}
Let $L \subseteq A^*$ be a language satisfying $\mc E_{ab=ba}$, $\mc E_{aab=abb}$ and $\mc E_{a=a.a}$, for every $a,b \in A$. Let us consider a finite colouring $\mc Q$ of $\mb N$ satisfying the condition from Corollary \ref{cor:reformulationequation}.
Then, for any words $w, w' \in A^*$ such that $w \sim_{\Q} w'$, we have
\[w \in L \Longleftrightarrow w' \in L.\]
\end{proposition}

\begin{proof}\label{proof:transductionN1u}
We consider a language satisfying this property, that is, $L \subseteq A^*$ such that there exists a finite colouring $\Q$ of $\mb N$ with $\ell$ colours, where $\ell \geq 1$, such that, for every $a,b \in A$,
\[ \bigcup_{i=1}^\ell A^{*}\otimes Q_i^2 \subseteq E_{L,f_{a,b},f_{b,a}},\]
\[ \bigcup_{i=1}^\ell A^{*}\otimes Q_i^3  \subseteq E_{L,f_{a,a,b},f_{a,b,b}},\]
and
\[ \bigcup_{i=1}^\ell  L_{Q_i} \otimes Q_i \subseteq E_{L,f_{a},f_{a}.a},\] 
by Corollary \ref{cor:reformulationequation}.
We consider two finite words $w,w' \in A^*$ such that $w \sim_{\Q} w'$.
We want to prove that $w$ is in $L$ if, and only if, $w'$ is in $L$. 

We can apply the reasoning we used in Proof \ref{proof:transductionN1ucasfacile} independently to each of the subsets $Q_1,...,Q_{\ell}$.
More precisely, for every $i \in \{1,...,\ell\}$, we consider
$\mc M_{\mc E,1,Q_{i}}$ as introduced in Proof \ref{proof:transductionN1ucasfacile} .
We can construct $\ell$ independent integers 
$(n_i)_{1 \leq i \leq \ell}$, 
$\ell$ independent families of pairs of maps
$((f^{i}_m,g^{i}_m)_{1 \leq m \leq n_i})_{1 \leq i \leq \ell}$ 
and $\ell$ independent finite sequence of words
$((v^{i}_m)_{1 \leq m \leq n_i})_{1 \leq i \leq l}$ 
satisfying the following properties: for every $i \in \{1,...,\ell\}$,
\begin{enumerate}
\item $\restriction{v^{i}_1}{Q_i} =\restriction{w}{Q_i} $ and $\restriction{v^{i}_{n_i}}{Q_i}=\restriction{w'}{Q_i}$;
\item for every $m \in \{1,...,n_i\}$, $(f^{i}_m,g^{i}_m) \in \mc M_{\mc E,1,Q_i}$;
\item for every $m \in \{1,...,n_i\}$, $f^{i}_{m}(v^{i}_m)=v^{i}_m$;
\item for every $m \in \{1,...,n_i-1\}$, $g^{i}_{m}(v^{i}_{m})=v^{i}_{m+1}$, and $g^{i}_n(v^{i}_{n_i})=v^{i}_{n_i}$.
\end{enumerate}

In particular,  \[w' = (\mathop{\bigcirc}\limits_{i = 1}^{\ell} \mathop{\bigcirc}\limits_{m = n_{i}}^{1} g^{i}_m)(w),\]
and thus by Proposition \ref{prop:transductionN1ucasfacile}, $w$ belongs to 
$L$ if, and only if, $w'$ belongs to $L$.
\end{proof}

We are now ready to prove that every language $L \subseteq A^*$ satisfying our equations is necessarily such that there exists a finite colouring $\Q$ of $\mb N$ such that $L$ saturates the finitely indexed equivalence relation $\sim_{\Q}$, and thus $L$ can be expressed as a Boolean combinations of the generators of $\mc B_1$.

\begin{corollary} \label{cor:completnessunaire}
Any language $L \subseteq A^*$ which satisfies the ultrafilter equations $\mc E_{ab=ba}$, $\mc E_{aab=abb}$ and $\mc E_{a=a.a}$, for every $a,b \in A$, is in $\mc B_1$.
\end{corollary}

\begin{proof}
Consider a language $L \subseteq A^*$ satisfying the equations $\mc E_{ab=ba}, \mc E_{aab=abb}$ and $\mc E_{a=a.a}$, for every $a,b \in A$.
By Proposition \ref{prop:reformulationEuv}, there exists a finite colouring $(Q_1,...,Q_{\ell})$ of $\mb N$, for some $\ell \geq 1$, such that, for every $a,b \in A$, 
\[ \bigcup_{i=1}^\ell A^{*} \otimes  Q_{i}^2 \subseteq E_{L,f_{a,b},f_{b,a}}, \]
\[ \bigcup_{i=1}^\ell A^{*} \otimes  Q_{i}^3 \subseteq E_{L,f_{a,a,b},f_{a,b,b}} \]
and
\[ \bigcup_{i=1}^\ell L_{Q_{i}} \otimes Q_{i} \subseteq E_{L,f_{a},f_{a}.a}.\]

Fix a word $w \in L$, then there exists a unique $\B \in \mc P(A)^\ell$ such that $w \in K_{\Q,\B}$.
By Proposition \ref{prop:transductionN1u}, for any $w' \in K_{\Q,\B}$, we have that $w' \in L$.
More generally, we deduce that the finitely indexed equivalence relation $\sim_{\mc Q}$ on $A^{*}$, is saturated by $L$. Therefore, there exists a finite family $\mathfrak{B} \subseteq \mc P(A)^{\ell}$ 
such that \[L = \bigcup_{\B \in \mathfrak{B}} K_{\mc Q,\B},\]
and by Proposition \ref{prop:caraB_kKQB}, we conclude that $L$ is in $\mc B_1$. 
\end{proof}

\section*{Conclusion}

In this article, we provided a detailed duality theoretic treatment of the Boolean algebra of languages corresponding to Boolean combinations of sentences written by using a block of $k$ existential quantifiers, letter predicates, and uniform numerical predicates. Several directions can be taken in order to continue this work. First, even though we provided a characterisation of the dual space corresponding to this Boolean algebra, one could try to make this characterisation even more explicit. The condition we give demands to consider every finite colouring of $\mb N^k$, which can make it a bit abstract to use in practice. We think that a result similar to Proposition \ref{prop:caraImc1} should exist in the case $k \geq 2$.
Second, note that, in this article, we restrict our attention to uniform numerical predicates. We do not say anything about the fragment $\mc B \Sigma [\mc N_k]$ defined by using arbitrary $k$-ary numerical predicates, nor about smaller fragments contained in $\mc B \Sigma [\mc N^{u}_k]$. It would be interesting to understand the duality theory for Boolean algebras obtained by applying a layer of existential quantifier to an arbitrary subalgebra of the Boolean algebra corresponding to quantifier-free formulas with $k$-ary uniform numerical predicates, and letter predicates.
Another question that could be asked, and which it is necessary to solve in order to understand $FO[\mc N]$, the fragment obtained by looking at every sentence that can be written in first-order logic on words, is how is it possible to take into account alternation of quantifiers. The main ingredient that allows us to conduct our study in this article is the fact that a block of existential quantifier commutes with finite disjunctions, this is what enables us to use the framework of modal algebra. However, whenever we consider an alternation of universal and existential quantifiers, this framework is no longer available and one has to think about a way to overcome this problematic.
We redirect the reader interested in this problem to a well written survey \cite{Zeitoun} on the state of the art on quantifier alternation hierarchy of first-order logic over words, in the case of regular languages, and why it is a difficult problem to solve.
We also refer to \cite{substitution} for an approach which is based on substitution of formulas and transductions.

\section*{Acknowledgments}
This was supported in part by the ERC grant DuaLL.

\bibliographystyle{unsrt}  
\bibliography{ref}

\end{document}